\documentclass[11pt, a4paper]{article}

\usepackage{amsmath,amsfonts,amsthm,amssymb}
\usepackage[retainorgcmds]{IEEEtrantools}
\usepackage[ruled,vlined,linesnumbered]{algorithm2e}
\usepackage{tikz}
\usepackage{kerkis}
\providecommand{\DontPrintSemicolon}{\dontprintsemicolon}
\usepackage{enumitem}
\usepackage{hyperref}
\usepackage[normalem]{ulem}
\usepackage{changepage}
\usepackage{microtype}
\usepackage[normalem]{ulem}
\usepackage{graphicx} 

\theoremstyle{plain}
\newtheorem{theorem}{Theorem}[section]
\newtheorem{fact}[theorem]{Fact}
\newtheorem{lemma}[theorem]{Lemma}

\newtheorem{definition}[theorem]{Definition}
\newtheorem{corollary}[theorem]{Corollary}

\theoremstyle{remark}

\newtheorem{remark}[theorem]{Remark}

\makeatletter
\newcommand*{\algotitle}[2]{%
	\stepcounter{algocf}%
	\hypertarget{algocf.title.\theHalgocf}{}%
	\NR@gettitle{#1}%
	\label{#2}%
	\addtocounter{algocf}{-1}%
}
\makeatother

\newenvironment{proofof}[1]
{\medskip\noindent{\it #1.}\hspace{1ex}}

\DeclareMathOperator{\opt}{\textsc{opt}}

\newcommand{\fopt}{\opt_{f}}

\DeclareMathOperator*{\argmax}{arg\,max}

\newcommand{\mysetminusD}{\hbox{\tikz{\draw[line width=0.6pt,line cap=round] (3pt,0) -- (0,6pt);}}}
\newcommand{\mysetminusT}{\mysetminusD}
\newcommand{\mysetminusS}{\hbox{\tikz{\draw[line width=0.45pt,line cap=round] (2pt,0) -- (0,4pt);}}}
\newcommand{\mysetminusSS}{\hbox{\tikz{\draw[line width=0.4pt,line cap=round] (1.5pt,0) -- (0,3pt);}}}
\newcommand{\mysetminus}{\mathbin{\mathchoice{\mysetminusD}{\mysetminusT}{\mysetminusS}{\mysetminusSS}}}
\allowdisplaybreaks

\usepackage[margin=1.2in]{geometry}

\title{On Budget-Feasible Mechanism Design\\ for Symmetric Submodular Objectives\thanks{A conference version appears in WINE 2017.}}
\author{
	Georgios Amanatidis 
	\and
	Georgios Birmpas 
	\and
	Evangelos Markakis \and 
	\normalsize{Athens University of Economics and Business, Department of Informatics.} \\
	\normalsize{\{gamana, gebirbas, markakis\}@aueb.gr}
}

\begin{document}
\maketitle

\begin{abstract}
We study a class of procurement auctions with a budget constraint, where an auctioneer is interested in buying resources or services from a set of agents. Ideally, the auctioneer would like to select a subset of the resources so as to maximize his valuation function, without exceeding a given budget. As the resources are owned by strategic agents however, our overall goal is to design mechanisms that are truthful,  budget-feasible, and obtain a good approximation to the optimal value. Budget-feasibility creates additional challenges, making several approaches inapplicable in this setting.  Previous results on budget-feasible mechanisms have considered mostly monotone valuation functions. In this work, we mainly focus on \textit{symmetric submodular} valuations, a prominent class of non-monotone submodular functions that includes \textit{cut functions}. We begin first with a purely algorithmic result, obtaining a $\frac{2e}{e-1}$-approximation for maximizing symmetric submodular functions under a budget constraint. We view this as a standalone result of independent interest, as it is the best known factor achieved by a deterministic algorithm. 
We then proceed to propose truthful, budget feasible mechanisms (both deterministic and randomized),  paying particular attention on the Budgeted Max Cut problem. Our results significantly improve the known approximation ratios for these objectives, while establishing polynomial running time for  cases where only exponential mechanisms were known. At the heart of our approach lies an appropriate combination of local search algorithms with results for monotone submodular valuations, applied to the derived local optima.
\end{abstract}

\section{Introduction}
\label{sec:intro}


We study a class of procurement auctions---also referred to as reverse auctions---with budget constraints. 
In a  procurement auction, an auctioneer is interested in buying goods or services from a set of agents. 
In this setting, selecting an agent comes at a cost and there is a hard budget constraint that should not be violated.
The goal of the auctioneer then is to select a budget-feasible subset of the agents so as to maximize his valuation function $v(\cdot)$, where $v(S)$ denotes the value derived when $S$ is the selected subset of agents to get services from.

The purely algorithmic version of the problem results in natural ``budgeted'' versions of known optimization problems. Since these problems are typically NP-hard, our focus is on approximation algorithms. Most importantly, in the setting considered here, the true cost of each agent is private information and we would like to design mechanisms that elicit truthful reporting by all agents. Hence, our ideal goal is to have truthful mechanisms that achieve a good approximation to the optimal value for the auctioneer, 
and are {\em budget feasible}, i.e., the sum of the payments to the agents does not exceed the prespecified budget. 
This framework of budget feasible mechanisms is motivated by recent application scenarios including crowdsourcing platforms, where agents can be viewed as workers providing tasks (e.g., \cite{AnariGN14,GoelNS14}), and influence maximization in networks, where agents correspond to influential users  (see e.g., \cite{Singer12,ABM16}, where the chosen objective is a coverage function). 

Budget feasibility makes the problem more challenging, with respect to truthfulness, as it already rules out well known mechanisms such as VCG. 
We note that the algorithmic versions of such problems often admit constant factor approximation algorithms. However, it is not clear how to appropriately convert them into truthful budget feasible mechanisms. 
The question is nontrivial even if we allow exponential time algorithms, since computational power does not necessarily make the problem easier (see the discussion in \cite{DobzinskiPS11}). All these issues create an intriguing landscape, where one needs to strike a balance between the incentives of the agents and the budget constraints.

The first positive results on this topic were obtained by Singer \cite{Singer10}, for the case where $v(\cdot)$ is an additive or a non-decreasing submodular function. Follow-up works provided refinements and further results for richer classes of functions (see the related work section). Most of these works, however, make the assumption that the valuation function is non-decreasing, i.e., $v(S)\leq v(T)$ for $S\subseteq T$, notable exceptions being the works of \cite{DobzinskiPS11} and \cite{BeiCGL12}.
Although monotonicity makes sense in several scenarios, one can think of examples where it is violated. E.g., \cite{DobzinskiPS11} studied the unweighted Budgeted Max Cut problem, as an eminent example of a non-monotone submodular objective function. 
Moreover, when studying models for influence maximization problems in social networks, adding more users to the selected set may some times bring negative influence \cite{BFO10} (some combinations of users may also not be compatible or well fitted together). Related to this, in the setting of \cite{ABM16}, an appealing class of objectives, under mild assumptions about the underlying network, correspond to weighted cut functions.
To further motivate the study of non-monotone submodular objectives, consider the following well-studied sensor placement problem \cite{Caselton84,Cressie93,KrauseSG08}: assume that we want to monitor some spatial phenomenon (e.g., the temperature of a specific environment), modeled as a Gaussian process. We may place sensing devices on some of the prespecified locations, but each location has an associated cost. A criterion for finding an optimal such placement, suggested by Caselton and Zidek \cite{Caselton84} for the unit cost case, is to maximize the {\it mutual information} between chosen and non chosen locations, i.e., we search for the subset of  locations that minimizes the uncertainty about the estimates in the remaining space. 
Such mutual information objectives are submodular but  not monotone. In addition, it is  straightforward to modify this problem to model participatory crowdsensing scenarios where users have incentives to lie about the true cost of installing a sensor.


It becomes apparent that we would like to aim for truthful mechanisms with good performance for subclasses of non-monotone functions. 
At the moment, the few results known for arbitrary non-monotone submodular functions have very large approximation ratios and often superpolynomial running time. Even worse, in most cases, we do not even know of deterministic mechanisms (see Table \ref{tab:results}).
In trying to impose more structure so as to have better positive results, there is an interesting observation to make: the examples that have been mentioned so far, i.e., cut functions and mutual information functions are \textit{symmetric submodular}, a prominent subclass of non-monotone submodular functions, where the value of a set $S$ equals the value of its complement. This subclass has received already considerable attention in operations research, see e.g., \cite{Fujishige83,Q98}, where more examples are also provided. 
We therefore find that symmetric submodular functions form a suitable starting point for the study of non-monotone functions.     
\medskip

\noindent {\bf Contribution:} 
The main focus of this work is on symmetric\footnote{In some works on mechanism design, symmetric submodular functions have a different meaning and refer to the case where $v(S)$ depends only on $|S|$. Here we have adopted the terminology of earlier literature on submodular optimization, e.g., \cite{Fujishige83}.} submodular functions. As suggested in \cite{Q98}, cut functions form a canonical example of this class. Consequently, we use the budgeted Max Cut problem throughout the paper as an illustrative example of how our more general approach could be refined for concrete objectives that have a well-behaved LP formulation. 
Our contributions can be summarized as follows:
\begin{itemize}[label={$\bullet$}]
	\item In Section \ref{sec:alg_ssm} we obtain a purely algorithmic result, namely  a $\frac{2e}{e-1}$-approximation for symmetric submodular functions under a budget constraint. We believe this is of independent interest, as it is the best known factor achieved by a deterministic algorithm (there exists already a randomized $e$-approximation) assuming only a value oracle for the objective function. 
	\item In Sections \ref{sec:exp_ssm} and \ref{sec:poly_ssm} we propose truthful, budget feasible mechanisms for arbitrary symmetric submodular functions, where previously known results regarded only randomized exponential mechanisms. We manage to significantly improve the known approximation ratios of such objectives by providing both randomized and deterministic mechanisms of exponential time. Moreover, we  propose a general scheme for producing constant factor approximation mechanisms that run in polynomial time when the objective function is well-behaved. These results provide partial answers to some of the open questions discussed in \cite{DobzinskiPS11}.
	\item  In the same sections, aside from studying arbitrary symmetric submodular functions, we also pay particular attention on the weighted and unweighted versions of Budgeted Max Cut. For the weighted version we obtain the first deterministic polynomial time mechanism with a 27.25-approximation 
	(where only an exponential randomized algorithm was known with a 768-approximation), while for the unweighted version  we improve the approximation ratio for polynomial randomized mechanisms, from 564 down to 10, and for polynomial deterministic mechanisms, from 5158 down to 27.25. 
	\item Finally, in Section \ref{sec:rand-xos} we briefly study the class of XOS functions, where we improve the current upper bound by a factor of 3.
\end{itemize}

All our contributions in mechanism design are summarized in Table \ref{tab:results}. We also stress that our mechanisms for general symmetric submodular functions use the value query model for oracle access to $v$, which is a much weaker requirement than the demand query model assumed in previous works, e.g., in \cite{DobzinskiPS11}.. 

Regarding the technical contribution of our work, the core idea of our approach is to exploit a combination of local search 
with  mechanisms for non-decreasing submodular functions. The reason local search is convenient for symmetric submodular functions is that it produces two local optima, 
and we can then prove that the function $v(\cdot)$ is non-decreasing within each local optimum. This allows us  to utilize mechanisms for non-decreasing submodular functions on the two subsets and then prove that one of the two solutions will attain a good approximation. The running time is still an issue under this approach, since finding an exact local optimum is not  guaranteed to terminate fast. However, even by finding \textit{approximate} local optima, 
within each of them the objective remains almost non-decreasing in a certain sense. This way we are still able to appropriately  adjust our mechanisms and obtain provable approximation guarantees. 
To the best of our knowledge, this is the first time that this \textit{``robustness under small deviations from monotonicity'}' approach is used to exploit known results for monotone objectives. 

\begin{table}[h]
	\begin{adjustwidth}{}{}
		\renewcommand{\arraystretch}{1.3}
		\setlength{\tabcolsep}{3pt}
		\centering
		{\small
		\begin{tabular}{c|c|c||c|c||c|c|||c|}
			\cline{2-8}
			& \multicolumn{2}{c||}{symmetric submod.} & \multicolumn{2}{c||}{unweighted cut}  & \multicolumn{2}{c|||}{weighted cut} & \multicolumn{1}{c|}{XOS}  \\ \cline{2-8} 
			& rand. & determ.  & rand. & determ. & rand. & determ.  &  rand. \\ \hline
			\multicolumn{1}{|c|}{{\footnotesize Known}} & $768\,^*\,^\dagger$ & --  & 564 & 5158  & $ 768\,^*\,^\dagger$ & -- & $768\,^*$   \\ \hline
			\multicolumn{1}{|c|}{{\footnotesize This paper}} & $10\, ^*$ & $10.90\, ^*$, {\scriptsize $(1\!+\!\rho)\Big(  2\!+\!\rho\! + \!\sqrt{\!\rho^2 \! +\! 4 \rho\! +1}\,\Big) $}  & $10$  & $27.25$  &  \multicolumn{2}{c|||}{$27.25$}  &  $244\,^*$ \\ \hline
		\end{tabular}
	}\smallskip
		\caption{A summary of our results on mechanisms. The asterisk ($*$) indicates that the corresponding mechanism runs in superpolynomial time. The dagger ($\dagger$) indicates that previously known results do not directly apply here; see Remark \ref{rem:n-sub} at the end of Section \ref{sec:prels}. The factor $\rho$ is an upper bound on the ratio of the optimal fractional solution to the integral one, assuming that we can find the former in polynomial time. The factor 768 is due to \cite{ChenGL11}, while the factors 564 and 5158 are due to \cite{DobzinskiPS11}.} \label{tab:results}
	\end{adjustwidth}
\end{table}

\noindent {\bf Related Work:} 
The study of budget feasible mechanisms, as considered here, was
initiated by Singer \cite{Singer10}. 
Later, Chen et al.~\cite{ChenGL11} significantly improved Singer's results, obtaining a randomized $7.91$-approximation mechanism and a deterministic $8.34$-approximation mechanism for non-decreasing submodular functions.
Several modifications of the mechanism by~\cite{ChenGL11} have been proposed that run in polynomial time for special cases \cite{Singer12,HorelIM14,ABM16}.
For subadditive functions, Dobzinski et al.~\cite{DobzinskiPS11} suggested a randomized $O(\log^2 n)$-approximation mechanism, and they gave the first constant factor mechanisms for non-monotone submodular objectives, specifically for \textit{cut functions}. The factor for subadditive functions was later improved to $O\big(\frac{\log n}{\log \log n}\big)$ by Bei et al.~\cite{BeiCGL12}, who also gave a randomized $O(1)$-approximation mechanism for XOS functions, albeit in exponential time, and further initiated the Bayesian analysis in this setting. 
An improved $O(1)$-approximation mechanism for XOS functions is also suggested in \cite{LeonardiMSZ16}. 
Recently, Amanatidis et al.~\cite{ABM16} suggested a $4$-approximation mechanism for a subclass of XOS functions that captures matroid and matching constraints.
There is also a line of related work under the {\em large market} assumption (where no participant can significantly affect the market outcome), which allows for polynomial time mechanisms with improved performance \cite{SinglaK13,AnariGN14,GoelNS14,BalkanskiH16,JT17}. 

On maximization of submodular functions subject to knapsack or other type of constraints, there is a vast literature, going back several decades, see, e.g., \cite{NemhauserWF78,Wolsey82}.  More recently Lee et al.~\cite{LeeMNS10} provided the first constant factor randomized algorithm for submodular maximization under $k$ matroid and  $k$ knapsack constraints, with factors $k+2+\frac{1}{k}$ and $5$ respectively. The problem was also studied by Gupta et al.~\cite{GuptaRST10} who  proposed a randomized algorithm, which achieves a $(4+\alpha)$-approximation\footnote{In the case of symmetric submodular functions the algorithm gives a deterministic 6-approximation.} in case of knapsack constraints, where $\alpha$ is the approximation guarantee of the unconstrained submodular maximization. Later on Chekuri et al.~\cite{ChekuriVZ14} suggested a randomized 3.07-approximation algorithm improving the previously known results. Finally, Feldman et al.~\cite{FeldmanNS11} and Kulik et al.~\cite{KulikST13} proposed their own randomized  algorithms when there are knapsack constraints, achieving an $e$-approximation.\footnote{The algorithm of \cite{KulikST13} can be derandomized without any performance loss, but only assuming an  additional oracle for the extension by expectation, say $V$, of the objective function $v$. When only an oracle for $v$ is available, estimation of $V$ by sampling is required in general.}


\section{Notation and Preliminaries}
\label{sec:prels}

We  use $A =[n]= \{1, 2, ..., n\}$ to denote a set of $n$ agents. Each agent $i$ is associated with a private cost $c_i$, denoting the cost for participating in the solution. 
We consider a procurement auction setting, where the auctioneer is equipped with a valuation function $v:2^{A}\to \mathbb{Q}_{\ge 0}$ and a  budget $B>0$. For $S\subseteq A$, $v(S)$ is the value derived by the auctioneer if the set $S$ is selected (for singletons, we will often write $v(i)$ instead of $v(\{i\})$).  Therefore, the algorithmic goal in all the problems we study is to select a set $S$ that maximizes $v(S)$ subject to the constraint $\sum_{i\in S} c_i \leq B$.
We assume oracle access to $v$ via value queries, i.e., we assume the existence of a polynomial time value oracle that returns $v(S)$ when given as input a set $S$.

We mostly focus on a natural subclass of submodular valuation functions that includes \textit{cut functions}, namely non-negative symmetric submodular functions.
Throughout this work we make the natural assumption that $v(\emptyset)=0$.


\begin{definition}
	A  function $v$, defined on $2^A$ for some set $A$, is \emph{submodular} if $v(S \cup \{i\}) - v(S) \geq v(T\cup \{i\}) - v(T)$ for any $S\subseteq T \subseteq A$, and $i\not\in T$. Moreover, it is \emph{non-decreasing} if $v(S) \le v(T)$ for any $S \subseteq T \subseteq A$, while it is \emph{symmetric} if $v(S) = v(A \mysetminus S) $ for any $S \subseteq A$.
\end{definition}

\noindent It is easy to see that $v$ cannot be both symmetric and non-decreasing unless it is a constant function. In fact, if this is the case 
and $v(\emptyset)=0$, then $v(S)=0$, for all $S\subseteq A$. 
We also state an alternative definition of a submodular function, which will be useful later on.

\begin{theorem}[Nemhauser et al.~\cite{NemhauserWF78}]\label{thm:alt_SM}
A set function $v$ is submodular if and only if for all $S, T \subseteq A$ we have
$v(T) \le v(S) + \sum\limits_{i \in T \mysetminus S} (v(S\cup\{i\}) - v(S)) - \sum\limits_{i \in S \mysetminus T} (v(S\cup T) - v(S\cup T \mysetminus \{i\}))$.
\end{theorem}

We often need to argue about optimal solutions of sub-instances, from an instance we begin with.
Given a cost vector $\mathbf{c}$, and a subset $X\subseteq A$, we denote by $\mathbf{c}_X$ the projection of $\mathbf{c}$ on $X$, and by $\mathbf{c}_{-X}$ the projection of $\mathbf{c}$ on $A\mysetminus X$.
We also let $\opt(X, v, \mathbf{c}_X, B)$ be 
the value of an optimal solution to the restriction of this instance on $X$,
i.e., $\opt(X, v, \mathbf{c}_X, B) = \max_{S:\,S\subseteq X,\, \mathbf{c}(S)\leq B}v(S)$.
Similarly, $\opt(X, v, \mathbf{c}_X, \infty)$ denotes the value of an optimal solution to the unconstrained version of the problem restricted 
on $X$.
For the sake of readability, we usually drop the valuation function and the cost vector, 
and  write $\opt(X, B)$ or $\opt(X, \infty)$.

Finally, in Sections \ref{sec:alg_ssm}-\ref{sec:poly_ssm} we make one further assumption: we assume that there is at most one item whose cost exceeds the budget. As shown in Lemma \ref{lem:Costs} and Corollary \ref{cor:Costs} in Appendix \ref{app:costs}, this is without loss of generality.
\medskip

\noindent{\bf Local Optima and Local Search.}
Given $v:2^{A}\to \mathbb{Q}$, a set $S\subseteq A$ 
is called a $(1+\epsilon)$-\textit{approximate local optimum} of $v$, if  
$(1+\epsilon) v(S) \ge v(S\mysetminus \{i\})$ and $(1+\epsilon) v(S) \ge v(S\cup \{i\})$  for every  $i\in A$. 
When $\epsilon = 0$, $S$ is called an \textit{exact local optimum} of $v$. Note that if $v$ is symmetric submodular, then $S$ is a $(1+\epsilon)$-approximate local optimum if and only if $A\mysetminus S$ is a $(1+\epsilon)$-approximate local optimum.
 
Approximate local optima  produce good approximations in unconstrained maximization of general submodular functions \cite{FeigeMV11}. However, here they are of interest for a quite different reason that becomes apparent in Lemmata  \ref{lem:almost-monotone} and \ref{lem:local-opt-monotone}. We can efficiently find approximate local optima using the local search algorithm \nameref{fig:alg-ls} of \cite{FeigeMV11}. Note that this is an algorithm for the unconstrained version of the problem, when there are no budget constraints.\medskip

\begin{algorithm}[H]
	\DontPrintSemicolon 
	\NoCaptionOfAlgo
	\algotitle{\textsc{Approx-Local-Search}}{fig:alg-ls.title}
	$S=\{i^*\}$, where $i^*\in\argmax_{i\in A}v(i)$ \;
	\While{there exists some $a$ such that $\max\{v(S\cup\{a\}), v(S\mysetminus\{a\})\} > (1+\epsilon/n^2) v(S)$}{
		\vspace{2 pt}\If{$v(S\cup\{a\}) > (1+\epsilon/n^2) v(S)$}{\vspace{2 pt}$S = S\cup\{a\}$}
		\Else{$S = S\mysetminus\{a\}$}}
	\Return $S$ \;
	\caption{\textsc{Approx-Local-Search}$(A, v, \epsilon)$ \ \cite{FeigeMV11}} \label{fig:alg-ls} 
\end{algorithm}\medskip 

If we care to find an exact local optimum, we can simply set $\epsilon=0$. In this case, however, we cannot argue about the running time of the algorithm  in general.

\begin{lemma}[inferred from \cite{FeigeMV11}]\label{lem:LS}
Given a submodular function $v:2^{[n]}\to \mathbb{Q}_{\ge 0}$ and a value oracle for $v$, 
\nameref{fig:alg-ls}$(A, v, \epsilon)$ outputs a $\left( 1+\frac{\epsilon}{n^2}\right) $-approximate local optimum using $O\left(\frac{1}{\epsilon}n^3 \log n \right)$ calls to the oracle.
\end{lemma}


\subsection{Mechanism Design}
In the strategic version that we consider here, every agent $i$ only has his true cost $c_i$ as private information.
A mechanism $\mathcal{M}=(f,p)$ in our context consists of an outcome rule $f$ and a payment rule $p$. Given a vector of cost declarations, 
$\mathbf{b} = (b_i)_{i \in A}$, where $b_i$ denotes the cost reported by agent 
$i$, the outcome rule of the mechanism selects the set $f(\mathbf{b})$. At the same time, it computes payments $p(\mathbf{b}) = (p_i(\mathbf{b}))_{i \in A}$ 
where $p_i(\mathbf{b})$ denotes the payment issued to agent $i$. Hence, the final utility of agent $i$ is $p_i(\mathbf{b}) - c_i$.

The main properties we want to ensure for our mechanisms  are the following. 
\begin{definition}
	A mechanism $\mathcal{M}=(f,p)$ is 
	\item[{\qquad \it 1.}] \emph{truthful}, if reporting $c_i$ is a dominant strategy for every agent $i$.
	\item[{\qquad \it 2.}] \emph{individually rational}, if $p_i(\mathbf{b})\geq 0$ for every $i\in A$, and $p_i(\mathbf{b}) \geq c_i$, for every $i\in f(\mathbf{b})$.
	\item[{\qquad \it 3.}] \emph{budget feasible}, if $\sum_{i\in A} p_i(\mathbf{b}) \leq B$ for every $\mathbf{b}$.
\end{definition}

For randomized mechanisms, we use the notion of \textit{universal truthfulness}, which means that the mechanism is a probability distribution over deterministic truthful mechanisms. 

We say that an outcome rule $f$ is {\em monotone}, if for every agent $i\in A$, and any vector of cost declarations $\mathbf{b}$, if $i\in f(\mathbf{b})$, then $i\in f(b_i', \mathbf{b}_{-i})$ for $b_i' \leq b_i$. This simply means that if an agent $i$ is selected in the outcome 
by declaring cost $b_i$, then by declaring a lower cost he should still be selected.
Myerson's lemma below implies that monotone algorithms admit truthful payment schemes.

\begin{lemma}[Myerson's lemma \cite{Myerson81}]\label{lem:myerson}
	Given a monotone algorithm $f$, there is a unique payment scheme $p$ such that $(f, p)$ is a truthful and individually rational  mechanism, given by
	\begin{displaymath}
	p_i(b)= \left\{ \begin{array}{ll}
	\sup_{b_i\in [c_i, \infty)} \{b_i: i\in f(b_i, b_{-i})\}\,, & \textrm{\emph{ if\ \  }} i\in f(b)\\
	0\,, & \textrm{\emph{ otherwise}}
	\end{array} \right.
	\end{displaymath}
\end{lemma}

The payments in Lemma \ref{lem:myerson} are often referred to as {\it threshold payments}, since they indicate the threshold at which an agent stops being selected. Myerson's lemma simplifies the design of truthful mechanisms by focusing only on constructing monotone algorithms and not having to worry about the payment scheme. Nevertheless, in the setting we study here budget feasibility clearly complicates things further. For all the algorithms presented in the next sections, we always assume that the underlying payment scheme is given by Myerson's lemma.
\medskip

\noindent{\bf Mechanisms for Non-Decreasing Submodular Valuations.}
In the mechanisms we design for non-mono\-tone submodular functions, we will repeatedly make use of truthful budget feasible mechanisms for non-decreasing submodular functions as subroutines. 
The best known such mechanisms 
are due to Chen et al.~\cite{ChenGL11}. 
Here, we follow the improved analysis of \cite{JT17} for the approximation ratio of the randomized mechanism \nameref{fig:alg-1a} of \cite{ChenGL11}, stated below. 
\medskip

\begin{algorithm}[H]
	\DontPrintSemicolon 
	\NoCaptionOfAlgo
	\algotitle{\textsc{Rand-Mech-SM}}{fig:alg-1a.title}
	Set $A'=\{i\ |\ c_i\le B\}$ and $i^*\in\argmax_{i\in A'}v(i)$ \;
	with probability $\frac{2}{5}$ \Return $i^*$ \;
	with probability $\frac{3}{5}$ \Return \textsc{Greedy-SM}$(A, v, \mathbf{c}, B/2)$\;
	\caption{\textsc{Rand-Mech-SM}$(A, v, \mathbf{c}, B)$ \ \cite{ChenGL11}}  
	\label{fig:alg-1a}
\end{algorithm} \medskip 

The mechanism \nameref{fig:alg-2} is a greedy algorithm that picks agents according to their ratio of marginal value over cost, given that this cost is not too large. For the sake of presentation, we assume the agents are sorted in descending order with respect to this ratio. The marginal value of each agent is calculated with respect to the previous agents in the ordering, i.e., $1=\argmax_{j\in A}\frac{v(j)}{c_j}$ and $i=\argmax_{j\in A\setminus [i-1]}\frac{v([j])-v([j-1])}{c_j}$ for $i\ge 2$.\medskip

\begin{algorithm}[H]
	\DontPrintSemicolon 
	\NoCaptionOfAlgo
	\algotitle{\textsc{Greedy-SM}}{fig:alg-2.title}
	Let $k=1$ and $S=\emptyset$ \;
	\While{$k\le |A|$ {\rm{\textbf {and}}} $v(S\cup\{k\}) > v(S)$ {\rm{\textbf {and}}} $c_k \le \frac{B}{2}\cdot \frac{v(S\cup\{k\})-v(S)}{v(S\cup\{k\})}$}{
		$S=S\cup\{k\}$ \;
		$k=k+1$
	}
	\Return $S$
	\caption{\textsc{Greedy-SM}$(A, v, \mathbf{c}, B/2)$ \ \cite{ChenGL11}} \label{fig:alg-2} 
\end{algorithm}  

\begin{lemma}[inferred from \cite{ChenGL11} and \cite{JT17}]\label{lem:GreedySM}
	\nameref{fig:alg-2} is monotone. Assuming a non-decreasing submodular function $v$ and the payments of Myerson's lemma, \nameref{fig:alg-2}$(A, v, \mathbf{c}, \allowbreak B/2)$ is budget-feasible and outputs a set $S$ such that $v(S)\ge \frac{1}{3} \cdot \opt(A,  B) - \frac{2}{3}\cdot v(i^*)$. 
\end{lemma}

A derandomized version of \nameref{fig:alg-1a.title} is also provided by \cite{ChenGL11}. It has all the desired properties, while suffering a small loss on the approximation factor. 
Here, following the improved analysis of \cite{JT17} for the ratio of \nameref{fig:alg-1a}, we  fine-tune this derandomized mechanism to obtain \nameref{fig:alg-1b-var} that has a better approximation guarantee. 
\smallskip
%

\begin{algorithm}[H]
	\DontPrintSemicolon 
	\NoCaptionOfAlgo
	\algotitle{\textsc{Mech-SM}}{fig:alg-1b-var.title}
	Set $A'=\{i\ |\ c_i\le B\}$ and $i^*\in\argmax_{i\in A'}v(i)$ \;
	\vspace{2pt} \If{$(2+\sqrt{6})\cdot v(i^*) \ge \opt(A\mysetminus\{i^*\},   B)$}{\vspace{2pt}\Return $i^*$}
	\Else{\Return \textsc{Greedy-SM}$(A, v, \mathbf{c}, B/2)$}
	\caption{\textsc{Mech-SM}$(A, v, \mathbf{c}, B)$} \label{fig:alg-1b-var} 
\end{algorithm}\medskip 

The next theorem summarizes the properties 
of \nameref{fig:alg-1a} and \nameref{fig:alg-1b-var}.


\begin{theorem}[inferred from \cite{ChenGL11}, \cite{JT17}, and Lemma \ref{lem:greedy_approximation}]\label{thm:mechSM}
\nameref{fig:alg-1a} runs in polynomial time, it is universally truthful, individually rational, budget-feasible, and has approximation ratio $5$. 
\nameref{fig:alg-1b-var} is deterministic, truthful, individually rational, budget-feasible, and has approximation ratio $3+\sqrt{6}$.
\end{theorem}

\begin{proof}
	Monotonicity (and thus truthfulness and individual rationality) and budget-feasibility of both mechanisms directly follow from \cite{ChenGL11}. What is left to show are the approximation guarantees.
	
	The key fact here is the following lemma from \cite{JT17}. Note that the lemma also follows from Lemma \ref{lem:greedy_approximation} for $\epsilon = 0 $ and $\beta = 0.5$.
	For the rest of this proof, $S$ will denote the outcome of \nameref{fig:alg-2}$(A, v, \mathbf{c}, B/2)$.
	
	\begin{lemma}[\cite{JT17}]\label{lem:greedy_approximation_jt}
		$\opt(A, B) \le 3\cdot v(S)+2\cdot v(i^*)$.
	\end{lemma}

	\noindent For \nameref{fig:alg-1a}, if $X$ denotes the outcome of the mechanism, then directly by Lemma \ref{lem:greedy_approximation_jt} we have 
	$\mathrm E(v(X))\ge \frac{3}{5}v(S)+\frac{2}{5}v(i^*)  \ge \frac{1}{5} \opt(A,  B)$,
	thus proving the approximation ratio. \smallskip
	
	\noindent For \nameref{fig:alg-1b-var} we consider two cases:
	
	If $i^*$ is returned by the mechanism, then 
	$(2+\sqrt{6}) \cdot v(i^*) \ge \opt(A\mysetminus\{i^*\},   B) \ge \opt(A,   B) - v(i^*)$,
	and therefore $\opt(A,   B) \le (3+\sqrt{6})\cdot v(i^*)$.
	
	On the other hand, if  $S$ is returned, then
	$(2+\sqrt{6}) \cdot v(i^*) < \opt(A\mysetminus\{i^*\},   B) \le \opt(A,   B)$. 
	Combining this with Lemma \ref{lem:greedy_approximation_jt} we have 
	$\opt(A,  B) \le 3\cdot v(S)+\frac{2}{2+\sqrt{6}}\opt(A,   B)$ and therefore $\opt(A,  B) \le \frac{3(2+\sqrt{6})}{\sqrt{6}}\cdot v(S) = (3+\sqrt{6}) \cdot v(S)$.
\end{proof}


\begin{remark} 
	\label{rem:n-sub}
	In  \cite{BeiCGL12} the proposed mechanisms regard  XOS and non-decreasing subadditive objectives, but it is stated that their results can be extended for general 
	subadditive functions as well. 
	This is achieved by  defining $\hat{v}(S)=\max_{T\subseteq S} v(T)$. It is easily seen that $\hat{v}$ is non-decreasing, 
	subadditive,
	and any solution that maximizes $v$ is also an optimal solution for $\hat{v}$. Although this is true for 
	subadditive  functions, it does not hold for submodular functions. In particular, if $v$ is submodular, then  $\hat{v}$ is not necessarily submodular. Therefore the results of \cite{ChenGL11} cannot be extended to our setting, even when time complexity is not an issue. An example of how $\hat{v}$ may fail to be submodular when $v$ is a cut function is given in Appendix \ref{app:counterexample}, where we also discuss how it can be derived from \cite{GuptaNS17} that in such cases $\hat{v}$ is actually XOS. This also implies that the best previously known approximation factor for the class of symmetric submodular functions was indeed 768 (inherited by the use of  $\hat{v}$).
	
\end{remark}


\section{The Core Idea: A Simple Algorithm for Symmetric Submodular Objectives}
\label{sec:alg_ssm}
This section deals with the algorithmic version of the problem: given a symmetric submodular function $v$, the  goal is to find $S\subseteq A$ that maximizes $v(S)$ subject to the constraint $\sum_{i\in S} c_i \leq B$. 
The main result is a deterministic $\frac{2e}{e-1}$-approximation algorithm for symmetric submodular functions. 
For this section only, the costs and the budget are assumed to be integral. 

Since our function is not monotone, we cannot directly apply the result of \cite{Sviridenko04}, which gives an optimal simple greedy algorithm for non-decreasing submodular maximization subject to a knapsack constraint. Instead, our main idea is to combine appropriately the result of \cite{Sviridenko04} with the local search used for unconstrained symmetric submodular maximization \cite{FeigeMV11}. 
At a high level, what happens is that 
local search produces an approximate solution $S$ for the unconstrained problem, and while this does not look related to our goal at first sight, $v$ is ``close to being non-decreasing'' on both $S$ and $A\mysetminus S$. This  becomes precise in Lemma \ref{lem:almost-monotone} below, but the point is that running a modification of the greedy algorithm of \cite{Sviridenko04}, 
on both $S$ and $A\mysetminus S$ will now produce at least one good enough solution. 
\smallskip

\begin{algorithm}[H]
	\DontPrintSemicolon 
	\NoCaptionOfAlgo
	\algotitle{\textsc{LS-Greedy}}{fig:alg-lsgreedy.title}
		$S = \nameref{fig:alg-ls}(A, v, \epsilon/4)$ \;
		$T_1 = \nameref{fig:alg-greedy-enum}(S, v, \mathbf{c}_{S}, B)$ \;
		$T_2 = \nameref{fig:alg-greedy-enum}(A\mysetminus S, v, \mathbf{c}_{A\mysetminus S}, B)$ \;
		Let $T$ be the best solution among $T_1$ and $T_2$ \;
		\Return $T$
	\caption{\textsc{LS-Greedy}$(A, v, \mathbf{c}, B, \epsilon)$} \label{fig:alg-lsgreedy} 
\end{algorithm}\smallskip 

The first component of our algorithm is the local search algorithm of \cite{FeigeMV11}. 
By Lemma \ref{lem:LS} and the fact that $v$ is symmetric, both $S$ and $A\mysetminus S$ are $\left( 1+\frac{\epsilon}{4n^2}\right) $-approximate local optima. We can now quantify the crucial observation that $v$ is close to being non-decreasing within the approximate local optima $S$ and $A\mysetminus S$. Actually, we only need this property on the local optimum that contains the best feasible solution.

\begin{lemma}\label{lem:almost-monotone}
	%
	Let $S$ be a $\left( 1+\frac{\epsilon}{4n^2}\right) $-approximate local optimum and consider 
	\[X\in \argmax_{Z\in\{S, A\mysetminus S\}}\opt(Z, B) \,.\]
	Then, for every $T\subsetneq X$ and every $i \in X\mysetminus T$, we have $v(T\cup\{i\}) - v(T) > - \frac{\epsilon}{n} \opt(X, B)$. 
\end{lemma}

Before proving Lemma \ref{lem:almost-monotone}, we begin with a simple fact and a useful lemma. We note that Fact \ref{fact:simple-facts} and Lemma \ref{lem:opt_vs_kopt} below require only subadditivity. Submodularity is used later, within the proof of Lemma \ref{lem:almost-monotone}.

\begin{fact}\label{fact:simple-facts}
	For any $S\subseteq A$,  $\max\{\opt(S,  B), \allowbreak \opt(A\mysetminus S, B)\}\ge 0.5 \opt(A, B)$
	since $\opt(S, B)+\opt(A\mysetminus S, B)\ge \opt(A, B)$ by subadditivity.
\end{fact}

\begin{lemma}\label{lem:opt_vs_kopt}
	For any $S\subseteq A$, $\opt(S, \infty)\le 2n \cdot \opt(A, B)$. 
\end{lemma}
\begin{proof}
	Recall that
	$|\{i\in A\ |\ c_i> B\}|\le 1$. Let $S^*\subseteq S$ be such that $v(S^*) = \opt(S, \infty)$. 
	By subadditivity
	we have
	$\opt(S, \infty) = v(S^*) \le \sum_{i\in S^*} v(i)$. Consider three cases. \medskip
	
	\noindent If $\{i\in A\ |\ c_i> B\}=\emptyset$, then by the fact that every singleton  is a feasible solution we have
	$\sum_{i\in S^*} v(i) \le n\cdot \max_{i\in A}v(i) \le n \cdot \opt(A, B)$.\medskip
	
	\noindent If $\{i\in A\ |\ c_i> B\}=\{x\}\nsubseteq S^*$, then every singleton in $A\mysetminus\{x\}$ is a feasible solution, and like before we have
	$\sum_{i\in S^*} v(i) \le (n-1)\cdot \max_{i\in A\mysetminus\{x\}}v(i) \le (n-1) \cdot \opt(A, B)$.\medskip
	
	\noindent If $\{i\in A\ |\ c_i> B\}=\{x\}\subseteq S^*$, then we need to bound $v(x)$. Since $v$ is symmetric we have $v(x)=v(A\mysetminus \{x\})\le \sum_{i\in A\mysetminus \{x\}} v(i) \le  (n-1) \cdot \max_{i\in A\mysetminus \{x\}}v(i)$. Therefore, 
	by using again that every singleton in $A\mysetminus\{x\}$ is a feasible solution, we have 
	$\sum_{i\in S^*} v(i) \le v(x) + (n-1)\cdot \max_{i\in A\mysetminus\{x\}}v(i)  \le (2n-2)\cdot \max_{i\in A\mysetminus\{x\}}v(i) \le 2n \cdot \opt(A, B)$.
\end{proof}

\begin{proofof}{Proof of Lemma \ref{lem:almost-monotone}}
	By Fact \ref{fact:simple-facts} we have $\opt(X, B) \ge 0.5 \opt(A, B)$. Let $T\subseteq X\mysetminus\{i\}$ for some $i\in X$. By submodularity we have $v(T\cup\{i\}) - v(T) \ge v(X) - v(X\mysetminus\{i\})$. Since $S$ is a $\left( 1+\frac{\epsilon}{4n^2}\right) $-approximate local optimum and $v$ is symmetric, $X$ is also a $\left( 1+\frac{\epsilon}{4n^2}\right)$-approximate local optimum. As a result, $v(X\mysetminus\{i\}) \le \left( 1+\frac{\epsilon}{4n^2}\right)  v(X)$ and thus $v(X) - v(X\mysetminus\{i\}) \ge - \frac{\epsilon}{4n^2} v(X) \ge - \frac{\epsilon}{4n^2} \opt(X, \infty) \ge - \frac{\epsilon}{2n} \opt(A, B)\ge - \frac{\epsilon}{n} \opt(X, B)$, where the second to last inequality follows from Lemma \ref{lem:opt_vs_kopt}.
	\qed
\end{proofof}\bigskip

The second component of \nameref{fig:alg-lsgreedy} is an appropriate modification of the greedy algorithm of \cite{Sviridenko04} for non-monotone submodular functions. It first  enumerates all solutions of size at most 3. Then, starting from each 3-set, it builds a greedy solution, and it outputs the best among these $\Theta(n^3)$ solutions. Here this idea is adjusted for non-monotone functions. \medskip

\begin{algorithm}[H]
	\DontPrintSemicolon 
	\NoCaptionOfAlgo
	\algotitle{\textsc{Greedy-Enum-SM}}{fig:alg-greedy-enum.title}
	{\small
	Let $S_1$ be the best feasible solution of cardinality at most 3 (by enumerating all such solutions) \;
	$S_2=\emptyset$ \;
	\For{every $U\subseteq A$ with $|U|=3$}{
	$T=U,\ t=1,\ A^0=A\mysetminus U$ \;
	\While{$A^{t-1} \neq \emptyset$}{
	Find $\theta_t = \displaystyle{\max_{i\in A^{t-1}} \frac{v(T\cup\{i\})-v(T)}{c_i}}$,  and let $i_t$ be an  element of $A^{t-1}$ that attains $\theta_t$ \;
	\If{$\theta_t\ge 0$\ \ and\ \ $\sum_{i\in T\cup\{i_t\}}c_i \le B$}{$T = T\cup\{i_t\}$}
	$A^t = A^{t-1}\mysetminus\{i_t\}$ \;
	$t=t+1$ \;
	}
	\If{$v(T)>v(S_2)$}{$S_2=T$}
	}
	Let $S$ be the best solution among $S_1$ and $S_2$ \;
	\Return $S$ 
	}
	\caption{\textsc{Greedy-Enum-SM}$(A, v, \mathbf{c}, B)$} \label{fig:alg-greedy-enum} 
\end{algorithm}\medskip 

By Fact \ref{fact:simple-facts}, at least one of $S$ and $A\mysetminus S$ contains a feasible solution of value at least $0.5 \opt(A, B)$.
Lemma \ref{lem:almost-monotone} guarantees that in this set, $v$ is very close to a non-decreasing submodular function. This is sufficient  for \nameref{fig:alg-greedy-enum} to perform almost as well as if $v$ was  non-decreasing.

%

\begin{theorem} \label{thm:alg_ssm}
For any $\epsilon >0$, algorithm \nameref{fig:alg-lsgreedy} achieves a $\left( \frac{2e}{e-1}+\epsilon\right)$-approximation. 
\end{theorem}

The proof is  deferred to Appendix \ref{app:alg_ssm}. 

Theorem \ref{thm:alg_ssm} suggests that a straightforward composition of two well known greedy algorithms achieves a good approximation for any symmetric submodular objective. We believe this is of independent interest and could be useful for other problems involving submodular optimization. 
From a mechanism design perspective, however, algorithm \nameref{fig:alg-lsgreedy} fails to be monotone and thus it cannot be used directly in the subsequent sections. In the next two sections, we remedy the problem by removing the enumeration part of the algorithm.  

\section{A First Take on Mechanism Design}
\label{sec:exp_ssm}
Utilizing the algorithmic approach of Section \ref{sec:alg_ssm} to get truthful mechanisms is not straightforward. One of the reasons is that \nameref{fig:alg-lsgreedy} is not monotone. We note that the algorithm \nameref{fig:alg-greedy-enum} without the enumeration part \textit{is}  monotone even for general objectives, but, to further complicate things, it is not guaranteed to be budget-feasible or have a good performance anymore.
Instead of computing approximate local optima 
like in Section \ref{sec:alg_ssm}, in this section we bypass most 
issues by computing exact local optima. The highlights of this simplified approach are polynomial mechanisms for unweighted cut functions with greatly improved guarantees. 

The price we have to pay, however, is that in general, finding exact local optima is not guaranteed to run in polynomial time \cite{SchafferY91}. Still, these mechanisms deepen our understanding of the problem. As mentioned in the Introduction, the problem seems to remain hard even when the running time is not an issue, and many existing mechanisms for various classes of functions are not polynomial. In particular, there are no better known mechanisms---even running in exponential time---for symmetric submodular objectives. We are going to deal further with the issue of running time in Section \ref{sec:poly_ssm}.

Below we give a randomized mechanism that reduces the known factor of 768 down to 10, as well as the first deterministic $O(1)$-approximation  mechanism for symmetric submodular objectives. In both mechanisms, local search produces a local maximum $S$ for the unbudgeted problem and then the budgeted problem is solved optimally on both $S$ and $A\mysetminus S$. As shown in Lemma \ref{lem:local-opt-monotone}, $v$ is non-decreasing on both $S$ and $A\mysetminus S$. Thus, running the mechanism
\nameref{fig:alg-1a} or \nameref{fig:alg-1b-var} of \cite{ChenGL11}, as described in Section \ref{sec:prels}, on $T\in \argmax_{X\in \{S, A\mysetminus S\}}\opt(X, B)$, directly implies a good solution. Since the resulting randomized and deterministic mechanisms are very similar, we state them together for succinctness. \smallskip


\begin{algorithm}[H]
	\DontPrintSemicolon 
	\NoCaptionOfAlgo
	\algotitle{\textsc{Rand-Mech-SymSM}}{fig:random-mech-symsm.title}
		$S = \nameref{fig:alg-ls}(A, v, 0)$ ~~ //find an exact local optimum \;
		\If{$\opt(S,  B)\ge \opt(A\mysetminus S,  B)$}{\Return \nameref{fig:alg-1a}$(S, v, \mathbf{c}_{S}, B)$ \quad(\textit{resp.}~\nameref{fig:alg-1b-var}$(S, v, \mathbf{c}_{S}, B)$)}
		\Else{\Return \nameref{fig:alg-1a}$(A\mysetminus S, v, \mathbf{c}_{A\mysetminus S}, B)$ \quad(\textit{resp.}~\nameref{fig:alg-1b-var}$(A\mysetminus S, v, \mathbf{c}_{A\mysetminus S}, B)$)}
	\caption{\textsc{Rand-Mech-SymSM}$(A, v, \mathbf{c}, B)$ \quad (\textit{resp.}~\textsc{Det-Mech-SymSM}$(A, v, \mathbf{c}, B)$)} \label{fig:random-mech-symsm} 
\end{algorithm}\smallskip 

The next simple lemma is crucial for the performance of both mechanisms for arbitrary submodular functions, and it shows how local search helps us exploit known results for non-decreasing submodular functions. Its proof is similar in flavor to the proof of Lemma \ref{lem:almost-monotone}. 

\begin{lemma} \label{lem:local-opt-monotone}
	Let $A$ be a set  and $v$ be an arbitrary submodular function defined on $2^A$. If $S$ is a local maximum of $v$, then $v$ is submodular and non-decreasing when restricted on $2^S$.
\end{lemma}

\begin{proof}
	The fact that $v$ is submodular when restricted on $2^S$ is trivial.
	Suppose now that the statement is not true and that $v$ is not non-decreasing on $2^S$. 
	That is, there exist $T, T' \subseteq S$ such that $T\subsetneq T'$ and $v(T)>v(T')$.
	Let $T'\mysetminus T = \{i_1, \ldots, i_r\}$. 
	
	By Theorem \ref{thm:alt_SM} we have 
	\[v(T) \le v(T') - \sum_{j\in[r]}(v(T') -v(T'\mysetminus \{i_j\})) \,, \]
	and therefore
	\[\sum_{j\in[r]}(v(T') -v(T'\mysetminus \{i_j\})) \le v(T')-v(T) <0 \,. \]
	We conclude that there is some $\ell\in[r]$ such that $v(T') -v(T'\mysetminus \{i_\ell\}) < 0$. 
	Then, by submodularity and the fact that $T'\mysetminus \{i_\ell\} \subseteq S\mysetminus \{i_\ell\}$, we get
	\[v(S) - v(S\mysetminus \{i_\ell\}) \le v(T') -v(T'\mysetminus \{i_\ell\}) < 0  \,.\]
	But then $v(S\mysetminus \{i_\ell\}) > v(S)$, which contradicts the fact that $S$ is a local maximum of $v$. So, it must be the case that $v$ is non-decreasing on the subsets of $S$.
\end{proof}

Since $v$ is symmetric, if $S$ is a local optimum, so is $A\mysetminus S$. Lemma \ref{lem:local-opt-monotone} suggests that we can use the mechanism \nameref{fig:alg-1a} (resp. \nameref{fig:alg-1b-var}) on $S$ and $A\mysetminus S$, to get the following implications. 

%
%

\begin{theorem} \label{thm:exp_ssm_det}
The mechanism \nameref{fig:random-mech-symsm} is universally truthful, individually rational, budget-feasible, and has approximation ratio $10$. The mechanism \textsc{Det-Mech-SymSM}  is deterministic, truthful, individually rational, budget-feasible, and has approximation ratio $6+2\sqrt{6}$. 
\end{theorem}

\begin{proof} 
	The fact that both mechanism \nameref{fig:random-mech-symsm} and \textsc{Det-Mech-SymSM} are  budget-feasible follows from the budget-feasibility of \nameref{fig:alg-1a} and \nameref{fig:alg-1b-var} respectively, established in \cite{ChenGL11}.
	
	For truthfulness and individual rationality, it suffices to show that  the allocation rule is monotone. Let us look first at \textsc{Det-Mech-SymSM}.
	Consider an agent $i$ with true cost $c_i$ in an instance $I$ where $i$ is included in the winning set. 
	Note first that the local search step is not affected by the costs, hence no player can influence the local optimum.
	Suppose now that it was the case that $\opt(S, B) \geq \opt(A\mysetminus S, B)$, hence $i\in S$.
	If player $i$ now declares a lower cost, then the optimal solution within $S$ can only get better, hence 
	the mechanism will
	run \nameref{fig:alg-1b-var} on $S$ as before. Since \nameref{fig:alg-1b-var} is monotone, player $i$ will again be selected in the solution.
	We conclude that the outcome rule is monotone and 
	\textsc{Det-Mech-SymSM} is truthful. 
	
	To prove that the randomized mechanism \nameref{fig:random-mech-symsm} is universally truthful
	we use similar arguments. 
	We fix the random bits of the mechanism and we consider a winning agent $i$ like before.
	Again, no player can influence the outcome of local search.
	Suppose it is the case that $\opt(S, B) \geq \opt(A\mysetminus S, B)$, hence $i\in S$.
	If  $i$  declares a lower cost, then the optimal solution within $S$ improves, and \nameref{fig:alg-1a} will still run  on $S$. Since \nameref{fig:alg-1a} is universally truthful, it is monotone given the random bits, and player $i$ will again be a winner.
	We conclude that 
	\nameref{fig:random-mech-symsm} is universally truthful.

	To argue now about the approximation ratio, suppose that we are in the case that $\opt(S, B) \geq \allowbreak \opt(A\mysetminus S, B)$ (the other case being symmetric). 
	We  know by Fact \ref{fact:simple-facts} that $\opt(S, B) \geq 0.5\cdot \opt(A, B)$. Hence, since we run either \textsc{Mech-SM} or \textsc{Rand-Mech-SM} on $S$, we will get twice their approximation ratio. This implies a ratio of $10$ for \nameref{fig:random-mech-symsm} and a ratio of $6+2\sqrt{6}$ for \textsc{Det-Mech-SymSM}.
\end{proof}

Lower bounds on the approximability have been obtained by Chen et al.~\cite{ChenGL11} for additive valuations.
Since additive functions are not 
symmetric, these lower bounds do not directly apply here. However, it is not hard to 
construct symmetric submodular functions that give the exact same bounds. 
\begin{lemma}
	\label{fact:lowerbounds}
	Independent of complexity assumptions, there is no deterministic (resp. randomized) truthful, 
	budget feasible mechanism that can achieve an approximation ratio better than $1+\sqrt{2}$ (resp. $2$).
\end{lemma}
\begin{proof} 
	The lower bounds of \cite{ChenGL11} are both for additive objectives (Knapsack). It suffices to show that given an instance $(A, v, \mathbf{c}, B)$ of Knapsack, we can construct an equivalent $(A', v', \mathbf{c}', B')$ instance of Budgeted Max Weighted Cut. This is straightforward. Consider a graph $G$ with vertex set $A\cup x$ and edge set $\{(i, x)\ |\ i\in A\}$.  For $i\in A$, vertex $i$ has cost $c'_i = c_i$, while vertex $x$ has cost $c'_x = B+1$. Edge $(i, x)$ has weight $v(i)$. Finally, for $S\subseteq A\cup x$,  $v'(S)$ is equal to the weight of the cut defined by $S$.
	
	This correspondence between items and vertices creates a natural correspondence between solutions.  It is clear that each feasible solution of the Knapsack instance  essentially defines  a feasible solution  to the Budgeted Max Weighted Cut instance of the same value and vice versa. In particular $\opt(A, v, \mathbf{c}, B) = \opt(A', v', \mathbf{c}', B)$. We conclude that any lower bound for Knapsack gives a lower bound for Budgeted Max Weighted Cut.
\end{proof}

Clearly, both mechanisms presented in this section require superpolynomial time in general (due to their first two lines), unless $P=NP$.
In both cases, instead of $\opt(\cdot, B)$ we could use the optimal solution of a fractional relaxation of the problem, at the expense of somewhat worse guarantees.
This does not completely resolve the problem, although this way local search becomes the sole bottleneck.
For certain objectives, however, we can achieve similar guarantees in polynomial time. Unweighted cut functions are the most prominent such example, and it is the focus of the next subsection.

\subsection{Polynomial Time Mechanisms for Unweighted Cut Functions}
\label{subsec:cut}

We begin with the definition of the problem when $v$ is a cut function: 
\smallskip

\noindent\emph{Budgeted Max Weighted Cut.} Given a complete graph $G$ with vertex set $V(G)=[n]$, non-negative weights $w_{ij}$ on the edges, non-negative costs $c_{i}$ on the nodes, and a positive budget $B$, find $X\subseteq [n]$ so that $v(X)=\sum_{i \in X}\sum_{j\in [n]\mysetminus X} w_{ij}$ is maximized subject to $\sum_{j\in X} c_j \le B$.\smallskip

For convenience, we assume the problem is defined on a complete graph as we can use zero weights to model any graph. In this subsection, we  focus on the unweighted version, where all weights are equal to either $0$ or $1$. We call this special case \emph{Budgeted Max  Cut}. The weighted version is considered in Subsection \ref{subsec:weight_cut}.

The fact that local search takes polynomial time to find an exact local optimum for the unweighted version \cite{KT06} does not suffice to make
\nameref{fig:random-mech-symsm} a polynomial time mechanism, since one still needs to compute $\opt(S, B)$ and $\opt(A\mysetminus S, B)$.
However, a small modification so that \nameref{fig:alg-1a}$(S, B)$ and \nameref{fig:alg-1a}$(A\mysetminus S, B)$ are returned with probability $1/2$ each (see Appendix \ref{app:bwmcut}) yields a randomized 10-approximate polynomial time mechanism. 
\medskip

\begin{algorithm}[H]
	\DontPrintSemicolon 
	\NoCaptionOfAlgo
	\algotitle{\textsc{Rand-Mech-UCut}}{fig:random-mech-ucut.title}
	$S = \nameref{fig:alg-ls}(A, v, 0)$ ~~ //find an exact local optimum \;
	with probability $1/2$ {\Return \nameref{fig:alg-1a}$(S, v, \mathbf{c}_{S}, B)$ }\;
	with probability $1/2$ {\Return \nameref{fig:alg-1a}$(A\mysetminus S, v, \mathbf{c}_{A\mysetminus S}, B)$} \;
	\caption{\textsc{Rand-Mech-UCut}$(A, v, \mathbf{c}, B)$} \label{fig:random-mech-ucut} 
\end{algorithm} 

\begin{theorem} \label{thm:poly_cut_rand}
	\nameref{fig:random-mech-ucut} is a randomized, universally truthful, individually rational, budget-feasible mechanism  for Budgeted Max  Cut that has approximation ratio $10$ and runs in polynomial time. 
\end{theorem}

\begin{proof} 
	Clearly, for the unweighted version of Max Cut, the mechanism runs in polynomial time. The fact that the mechanism is truthful, individually rational and budget feasible follows from the same arguments as in the proof of Theorem \ref{thm:exp_ssm_det}. \smallskip 

	Finally, let 
	\begin{IEEEeqnarray*}{l}
		X_S = \nameref{fig:alg-2}(S, v, \mathbf{c}_S, B/2), \ \ 
		X_{A\mysetminus S} = \nameref{fig:alg-2}(A\mysetminus S, v, \mathbf{c}_{A\mysetminus S}, B/2), \\ 
		i_S\in \allowbreak \argmax_{\{i\in S\, |\, c_i\le B\}}v(i),\quad \text{and} \quad i_{A\mysetminus S}\in \argmax_{\{i\in {A\mysetminus S}\, |\, c_i\le B\}}v(i)\,.
	\end{IEEEeqnarray*}
	
	Since we run mechanism \textsc{Rand-Mech-SM} with probability $1/2$ on $S$ and $A\mysetminus S$, it is not hard to see that each of $X_S$ and $X_{A\mysetminus S}$ is returned with probability 3/10, while each of $i_S$ and $i_{A\mysetminus S}$ is returned with probability 2/10. If $X$ denotes the outcome of the mechanism, then using subadditivity and Lemma \ref{lem:greedy_approximation_jt}, we get 
	\begin{IEEEeqnarray*}{rCl}
		\mathrm E(v(X)) & = & \frac{3}{10}v(X_S)+\frac{2}{10}v(i_S) +\frac{3}{10}v(X_{A\mysetminus S})+\frac{2}{10}v(i_{A\mysetminus S})  \\
		& \ge & \frac{1}{10} \opt(S,  B) + \frac{1}{10} \opt(A\mysetminus S,  B) \\
		& \ge & \frac{1}{10} \opt(A,  B) \,,
	\end{IEEEeqnarray*}
	thus establishing an approximation ratio of $10$.
\end{proof}


In order to design deterministic mechanisms that run in polynomial time, we first optimize a deterministic mechanism of \cite{ABM16} to obtain \nameref{fig:alg-1b-frac} below, which is applicable for non-decreasing submodular functions. The difference with  \nameref{fig:alg-1b-var} from Section \ref{sec:prels}, is that now we assume that 
a fractional relaxation of our problem can be solved optimally and that the fractional optimal solution is within a constant of the integral solution. To do this, we will consider a LP relaxation of an ILP formulation of the problem, and also bound the solution, using pipage rounding.
In particular, let $v(\cdot)$ be a non-decreasing submodular function, $A'=\{i\in A\ |\ c_i\le B\}$, and consider a relaxation of our problem for which we have an exact algorithm. Moreover, suppose that $\fopt(A', v, \mathbf{c}_{A'}, B) \le \rho\cdot\opt(A', v, \mathbf{c}_{A'}, B) = \rho\cdot\opt(A, v, \mathbf{c}, B)$ for any instance,  where  $\fopt$ and $\opt$ denote the value of an optimal solution to the relaxed and the original problem respectively.

\begin{theorem}[inferred from \cite{ABM16} and Lemma \ref{lem:greedy_approximation}]\label{thm:mechSM-frac}
	\nameref{fig:alg-1b-frac} is deterministic, truthful, individually rational, budget-feasible, and has approximation ratio $\rho + 2 +\sqrt{\rho^2 +4\rho +1}$. Also, it runs in polynomial time as long as the exact algorithm for the relaxed problem runs in polynomial time.
\end{theorem}

\begin{algorithm}[H]
	\DontPrintSemicolon 
	\NoCaptionOfAlgo
	\algotitle{\textsc{Mech-SM-frac}}{fig:alg-1b-frac.title}
	Set $A'=\{i\ |\ c_i\le B\}$ and $i^*\in\argmax_{i\in A'}v(i)$ \;
	\vspace{2pt} \If{$\left( \rho +1 +\sqrt{\rho^2 +4\rho +1}\right) \cdot v(i^*) \ge \fopt(A'\mysetminus\{i^*\},   B)$}{\vspace{2pt}\Return $i^*$}
	\Else{\Return \textsc{Greedy-SM}$(A, v, \mathbf{c}, B/2)$}
	\caption{\textsc{Mech-SM-frac}$(A, v, \mathbf{c}, B)$} \label{fig:alg-1b-frac} 
\end{algorithm}\smallskip 

Note that when the relaxed problem is the same as the original ($\rho = 1$), then \nameref{fig:alg-1b-frac} becomes \nameref{fig:alg-1b-var} and the two approximation ratios coincide.
A discussion on how different results combine into Theorem \ref{thm:mechSM-frac} is deferred to Appendix \ref{app:bumcut}. 

Hence, to obtain a deterministic mechanism for Budgeted Max Cut, we will use an LP-based approach, and we will run \nameref{fig:alg-1b-frac} on an appropriate local maximum. For this, we will first need to compare the value of an optimal solution of a fractional relaxation to the value of an optimal solution of the original problem. 
Ageev and Sviridenko \cite{AgeevS99} studied a different Max Cut variant, but we follow a similar approach to obtain the desired bound for our problem as well. We begin with a linear program formulation of the problem. 
Our analysis is carried out for the weighted version of the problem, as we are going to reuse some  results in Subsection \ref{subsec:weight_cut}, which deals with  weighted cut functions. To be more precise, we want to argue about a variant of the problem where we may only be allowed to choose the solution from a specified subset of $A$. That is, we formulate below the sub-instance $I = (X, v, \mathbf{c}_X, B)$ of $(A, v, \mathbf{c}, B)$, where in fact $X \subseteq A'=\{i\in A\ |\ c_i\le B\}$.

We associate a binary variable $x_i$ for each vertex $i$, and the partition of $A=[n]$ according to the value of the $x_i$s defines the cut. There is also a binary variable $z_{ij}$ for each edge $\{i, j\}$ which is the indicator variable of whether $\{i, j\}$ is in the cut.
\begin{IEEEeqnarray}{rCll}
	\mbox{maximize:\ \ } & & \sum_{i\in [n]}\sum_{j\in [n]\mysetminus[i]} w_{ij} z_{ij} & \label{eq:bwmcut1}\\
	\mbox{subject to:\ \ } & & z_{ij}\le x_i +x_j\ ,\qquad & \forall i\in [n], \forall j \in [n] \mysetminus [i]\\
	& & z_{ij}\le 2- x_i -x_j\ ,\qquad & \forall i\in [n], \forall j \in [n] \mysetminus [i]\\
	& & \sum_{i\in [n]} c_i x_i \le B &  \\
	& & x_i = 0\ ,\qquad & \forall i\in [n]\mysetminus X\\
	& & 0\le x_i, z_{ij} \le 1\ ,& \forall i\in [n],\, \forall j \in [n] \mysetminus [i] \label{eq:bwmcut5}\\
	& & x_i\in \{0,1\}\ ,& \forall i\in [n] \label{eq:bwmcut6}
\end{IEEEeqnarray} 

It is not hard to see that \eqref{eq:bwmcut1}-\eqref{eq:bwmcut6} is a natural ILP formulation for Budgeted Max Weighted Cut and \eqref{eq:bwmcut1}-\eqref{eq:bwmcut5} is its linear relaxation. Let $\opt(I)$ and ${\fopt}(I)$ denote the optimal solutions to \eqref{eq:bwmcut1}-\eqref{eq:bwmcut6} and \eqref{eq:bwmcut1}-\eqref{eq:bwmcut5} respectively for instance $I$. 
To show how these two are related we will use the technique of pipage rounding \cite{AgeevS99,AgeevS04}. Although we do not provide a general description of the  
technique, the proof of the next theorem is self-contained (see Appendix \ref{app:bumcut}). 

\begin{theorem}\label{lem:bwmcut}
	Given the fractional relaxation \eqref{eq:bwmcut1}-\eqref{eq:bwmcut5} for Budgeted Max Weighted Cut, we have that 
	$\fopt(I) \le (2+2\beta_I)\cdot \opt(I)$,
	for any instance $I$, where $\beta_I$ is such that $\max_{i\in A'}v(i) \le \beta_I \cdot \opt(I)$.
\end{theorem}

Note that there always exists some $\beta_I \leq 1$ for every instance $I$, hence, the above theorem implies a worst-case upper bound of $4$.
Now, we may modify \textsc{Det-Mech-SymSM} to use $\fopt$ instead of $\opt$, and \nameref{fig:alg-1b-frac} instead of \nameref{fig:alg-1b-var}.
This results in the following deterministic mechanism that runs in polynomial time. \smallskip

\begin{algorithm}[H]
	\DontPrintSemicolon 
	\NoCaptionOfAlgo
	\algotitle{\textsc{Det-Mech-UCut}}{fig:det-mech-ucut.title}
		Set $A'=\{i\ |\ c_i\le B\}$ and $i^*\in\argmax_{i\in A'}v(i)$ \;
		\vspace{2 pt}\If{$26.25\cdot v(i^*) \ge \fopt(A'\mysetminus\{i^*\},   B)$ \label{line:cut1}}{\vspace{2pt}\Return $i^*$ \label{line:cut2}}
		\Else{$S = \nameref{fig:alg-ls}(A, v, 0)$ \;
			\If{$\fopt(S \cap A',  B)\ge \fopt(A'\mysetminus S,  B)$ \label{line:cut5}}{\vspace{2 pt}\Return \nameref{fig:alg-1b-frac}$(S, v, \mathbf{c}_{S}, B)$ \label{line:cut7}}
			\Else{\Return \nameref{fig:alg-1b-frac}$(A\mysetminus S, v, \mathbf{c}_{A\mysetminus S}, B)$ \label{line:cut9}}
		}
	\caption{\textsc{Det-Mech-UCut}$(A, v, \mathbf{c}, B)$} \label{fig:det-mech-ucut} 
\end{algorithm}\smallskip 

\begin{theorem} \label{thm:poly_cut_det}
	\nameref{fig:det-mech-ucut} is a deterministic, truthful, individually rational, budget-feasible mechanism for Budgeted Max  Cut that has approximation ratio $27.25$ and runs in polynomial time.
\end{theorem}

\begin{proof} 
	Clearly the mechanism runs in polynomial time.
	
	For truthfulness and individual rationality, it suffices to show that  the allocation rule is monotone, i.e., a winning agent $j$ remains a winner if he decreases his cost to be $c_j' < c_j$. If $j=i^*$ and he wins in line \ref{line:cut1}, then his bid is irrelevant and he remains a winner. If $j\neq i^*$, or $j=i^*$ but he wins after line \ref{line:cut2}, we may assume he wins at line \ref{line:cut7} (the case of line \ref{line:cut9} is symmetric). When bidding $c_{j}'<c_{j}$, the decision of the mechanism in line \ref{line:cut1} does not change (if $j\neq i^*$ then $\opt_f$ is improved, if $j= i^*$ nothing changes). Further, since he cannot influence the local search,  $j$ is still in $S$ and \nameref{fig:alg-1b-frac} is executed. By the monotonicity of \nameref{fig:alg-1b-frac} we have that $j$ is still a winner.
	Therefore, the mechanism is monotone.
	
	If the winner is $i^*$ in line \ref{line:cut2}, then his payment is $B$. Otherwise, budget-feasibility follows from the budget-feasibility of \nameref{fig:alg-1b-frac} and the observation that the comparison in line \ref{line:cut1} only gives additional upper bounds on the payments of winners from \nameref{fig:alg-1b-frac}. 
	
	It remains to prove the approximation ratio. We consider two cases. Let $\alpha = 26.25$.
	If $i^*$ is returned in line \ref{line:cut2}, then 
	\[\alpha \cdot v(i^*) \ge \opt_f(A'\mysetminus\{i^*\},  B) \ge \opt(A'\mysetminus\{i^*\},  B) = \opt(A\mysetminus\{i^*\},  B) \ge \opt(A, B) - v(i^*) \,,\]
	and therefore $\opt(A,  B)\le (\alpha +1) \cdot v(i^*) = 27.25 \cdot v(i^*)$. 
	
	On the other hand, if  $X= \nameref{fig:alg-1b-frac}(S, B)$ is returned in line \ref{line:cut7}, then by Theorem \ref{lem:bwmcut} with factor 4 we have
	\[\alpha \cdot v(i^*) < \opt_f(A'\mysetminus\{i^*\}, B) \le 4 \cdot \opt(A'\mysetminus\{i^*\}, B) = 4 \cdot \opt(A\mysetminus\{i^*\}, B) \le 4 \cdot \opt(A, B).\] Therefore, $v(i^*) <\frac{4}{\alpha} \opt(A,   B)$ and for the remaining steps of the mechanism  we can use  Theorem \ref{lem:bwmcut} with factor $2+8/\alpha$. 
	
	At line \ref{line:cut5} it must be the case that $\fopt(S\cap A', B)\ge \fopt(A'\mysetminus S, B)$. Thus,
\begin{IEEEeqnarray*}{rCl}
\Big( 2+\frac{8}{\alpha}\Big) \opt(S, B)  & = & \Big( 2+\frac{8}{\alpha}\Big) \opt(S\cap A', B) \ge \fopt(S\cap A', B)\ge \fopt(A'\mysetminus S, B)  \\
 & \ge & \opt(A'\mysetminus S, B) = \opt(A\mysetminus S, B) \ge \opt(A, B) - \opt(S, B)\,.
\end{IEEEeqnarray*}
	Therefore $\opt(S, B) \ge \frac{\alpha}{3\alpha+8} \opt(A, B)$. By Theorem \ref{thm:mechSM-frac} we have \[\left( 4+8/\alpha +\sqrt{(2+8/\alpha)^2 +32/\alpha +9}\right) v(X)\ge \opt(S, B) \ge \frac{\alpha}{3\alpha+8} \opt(A, B)\,,\]
	and by substituting $\alpha$ and doing the calculations we get $\opt(A,  B)\le 27.25 \cdot v(X)$.
	
	The case where $X= \nameref{fig:alg-1b-frac}(A\mysetminus S,  B)$ is returned in line \ref{line:cut9} is symmetric to the case above and it need not be considered separately. We conclude that  $\opt(A,  B)\le 27.25 \cdot \nameref{fig:det-mech-ucut}(A, B)$.
\end{proof}

\section{Symmetric Submodular Objectives Revisited}
\label{sec:poly_ssm}

Suppose that for a symmetric submodular function $v$, an optimal fractional solution can be found efficiently and that $\fopt(A', B) \le \rho\cdot\opt(A, B)$ for any instance,  where  $\fopt$ and $\opt$ denote the value of an optimal solution to the relaxed and the original problem respectively, and $A'=\{i\in A\ |\ c_i\le B\}$. 

A natural question is whether the approach taken for unweighted cut functions can be fruitful for other symmetric submodular objectives. In the mechanisms of Subsection \ref{subsec:cut}, however, the complexity of local search can be a bottleneck even for objectives where an optimal fractional solution can be found  fast and it is not far from the optimal integral solution. So, we now return to the idea of Section \ref{sec:alg_ssm}, where local search runs in polynomial time and produces an approximate local maximum; unfortunately, the nice property of monotonicity in each side of the partition (Lemma \ref{lem:local-opt-monotone}) does not hold any longer. 


This means that the approximation guarantees of such mechanisms do not follow in any direct way from existing work. Moreover, budget-feasibility turns out to be an even more delicate issue 
since it crucially depends on the (approximate) monotonicity of the valuation function. Specifically, when 
a set $X$ only contains a very poor solution to the original problem, every existing proof of budget feasibility for the restriction of $v$ on $X$ completely breaks down. 
Since we cannot expect that  an approximate local maximum $S$ and its complement $A\mysetminus S$ both contain a ``good enough'' solution to the original problem, we need to make sure that \nameref{fig:alg-2} never runs on the wrong set. 

The mechanism \nameref{fig:det-mech-ucut} for the unweighted cut problem seems to take care of this and we are going to build on it, in order to propose mechanisms for arbitrary symmetric submodular functions. To do so we replace the constant $26.25$ that appears in \nameref{fig:det-mech-ucut} by $\alpha = (1+\rho)\left( 2+\rho + \sqrt{\rho^2 +4 \rho +1}\right) -1$  
and we find an approximate local maximum instead of an exact local maximum. Most importantly, in order to achieve budget-feasibility we introduce a  modification of \nameref{fig:alg-1b-frac} (which we call \nameref{fig:alg-1b-frac-var}, see Appendix \ref{app:poly_ssm_proofs}) that runs \nameref{fig:alg-2} with a slightly reduced budget. The parameter $\epsilon'$ that appears in the description of the  mechanism below  is determined by the analysis of the mechanism and only depends on the constants $\rho$ and $\epsilon$.
\smallskip

\begin{algorithm}[H]
	\DontPrintSemicolon 
	\NoCaptionOfAlgo
	\algotitle{\textsc{Det-Mech-SymSM-frac}}{fig:det-mech-symsm-frac-xx.title}
		Set $A'=\{i\ |\ c_i\le B\}$ and $i^*\in\argmax_{i\in A'}v(i)$ \;
		\vspace{2 pt}\If{$\alpha \cdot v(i^*) \ge \fopt(A'\mysetminus\{i^*\},   B)$ }{\vspace{2pt}\Return $i^*$ } 
		\Else{$S = \nameref{fig:alg-ls}(A, v, \epsilon')$ \;
			\If{$\fopt(S\cap A',   B)\ge \fopt(A'\mysetminus S,  B)$ }{\vspace{2 pt}\Return \nameref{fig:alg-1b-frac-var}$(S, v, \mathbf{c}_{S}, B, (1-(\alpha_1+2) \epsilon'))$}
			\Else{\Return \nameref{fig:alg-1b-frac-var}$(A\mysetminus S, v, \mathbf{c}_{A\mysetminus S}, B, (1-(\alpha_1+2) \epsilon'))$}
		}
	\caption{\textsc{Det-Mech-SymSM-frac}$(A, v, \mathbf{c}, B, \epsilon)$} 
\end{algorithm}\smallskip 


Theorem \ref{thm:poly_ssm_det} below shows that for any objective for which we can establish a constant upper bound $\rho$ on the ratio of the fractional and the integral optimal solutions,
we have constant factor approximation mechanisms that run in polynomial-time .

\begin{theorem} \label{thm:poly_ssm_det}
	For any $\epsilon >0$, \nameref{fig:det-mech-symsm-frac} is a deterministic, truthful, individually rational, budget-feasible mechanism for symmetric submodular valuations, that has approximation ratio $\alpha +1 +\epsilon$ and runs in polynomial time.
\end{theorem}

Appendix \ref{app:bwmcut} is mostly dedicated to the proof of the theorem. We view Theorem \ref{thm:poly_ssm_det} as the most technically demanding result of this work. There are several steps involved in the proof, since we need the good properties of \nameref{fig:alg-2} to still hold 
even for objectives that are not exactly non-decreasing.

\subsection{Weighted Cut Functions}
\label{subsec:weight_cut}
Let us return now to the Max Cut problem, and consider the weighted version. 
An immediate implication of Theorem 
\ref{thm:poly_ssm_det} is that we get a deterministic polynomial-time mechanism for Budgeted Max Weighted Cut with approximation ratio 58.72. This is just the result of substituting $\rho = 4$, as suggested by Theorem \ref{lem:bwmcut}, in the formula for $\alpha$. 

\begin{corollary}\label{cor:poly_wcut_rand}
	There is a deterministic, truthful, individually rational, budget-feasible mechanism for Budgeted Max Weighted Cut that has approximation ratio 58.72 and runs in polynomial time.
\end{corollary}

However, Theorem \ref{lem:bwmcut} says something stronger: given that  $\max_{i\in A'}v(i)$ is small compared to $\opt(A, B)$, $\rho$ is strictly smaller than 4. Note that the first step in \nameref{fig:det-mech-symsm-frac}  is to compare $\max_{i\in A'}v(i)$ to $\fopt(A'\mysetminus\{i^*\},  B)$. This implies an upper bound on $\max_{i\in A'}v(i)$ in the following steps and we can use it to further fine-tune our mechanism. In particular, by setting $\alpha = 26.245$ instead of $(1+4)\left( 2+4 + \sqrt{16 +16 +1}\right) -1 = 57.72$ in \nameref{fig:det-mech-symsm-frac}, we prove in Appendix \ref{app:poly_ssm_proofs} 
the following improved result that matches the approximation guarantee for unweighted cut functions.

%

\begin{theorem} \label{thm:poly_wcut_det}
	There is a  deterministic, truthful, individually rational, budget-feasible mechanism for Budgeted Max Weighted Cut that has approximation ratio 27.25, and runs in polynomial time. 
\end{theorem}

\section{An Improved Upper Bound for XOS Objectives}
\label{sec:rand-xos}

In \cite{BeiCGL12} a randomized, universally truthful and budget-feasible 768-approximation mechanism was introduced for XOS functions.
For several of the results in Sections \ref{sec:exp_ssm} and \ref{sec:poly_ssm} the best previously known upper bound follows from this result (see also Remark \ref{rem:n-sub}). 
In this section we slightly modify their mechanism to improve its performance.\footnote{In a recent manuscript, Leonardi et al.~\cite{LeonardiMSZ16} also suggest a fine-tuning of this mechanism that yields a factor 436.} 
Although there is not much novelty in this result, it feels appropriate to provide this tighter analysis, given that the factor of 768 has been the benchmark against which our results are presented.

We begin with the definition of XOS functions.

\begin{definition}
A valuation function, defined on $2^A$ for some set $A$, is \emph{XOS} or \emph{fractionally subadditive}, if there exist non-negative additive functions $\alpha_1,\alpha_2, ...,\alpha_r$, for some finite $r$, such that 
$v(S) = \max \{\alpha_1(S), \alpha_2(S), ...,\allowbreak\alpha_r(S) \}$.
\end{definition}

Note that above we define non-decreasing XOS functions, and everything is stated and proved for those. However, there is a relatively straightforward way to extend any result to general XOS functions (as defined in \cite{GuptaNS17}); see Remark \ref{rem:n-sub} and Appendix \ref{app:counterexample}. 

Below we provide the mechanism \nameref{fig:XOS-mechanism-main-r}. 
Our mechanism has the same structure as the one presented in \cite{BeiCGL12} but we tune its parameters and perform a slightly different analysis in order to improve the approximation factor. 
The mechanism \textsc{Additive-Mechanism} of \cite{ChenGL11} for additive valuation functions is used as a subroutine. \textsc{Additive-Mechanism} is a universally truthful, 3-approximate mechanism (see Theorem B.2 in \cite{ChenGL11}).
Initially, we revisit the \textit{random sampling} part of the mechanism and modify the threshold bound of line \ref{line:xos_2}:
\medskip 

\begin{algorithm}[H]
	\DontPrintSemicolon 
	\NoCaptionOfAlgo
	\algotitle{\textsc{Sample-XOS}}{fig:XOS-random-sample-r.title}
	Pick each item independently at random with probability $\frac{1}{2}$ to form a set $T$ \;
	Compute $\opt(T, v, \mathbf{c}_T, B)$ and 
	set a threshold $t=\frac{\opt(T, v, \mathbf{c}_T, B)}{4.6 B}$ \; \label{line:xos_2}
	Find a set $S^*\in \argmax_{S\subseteq A\mysetminus T}\{v(S)-t\cdot \mathbf c(S)\}$ \;
	Find an additive function $\alpha$ in the XOS representation of $v(\cdot)$ with $\alpha(S^*)=v(S^*)$  \;
	\Return \textsc{Additive-Mechanism}$(S^*, \alpha, \mathbf{c}_{S^*}, B)$ 
	\caption{\textsc{Sample-XOS}} \label{fig:XOS-random-sample-r} 
\end{algorithm}\medskip 

This part is used as one of the two alternatives of the main mechanism. We modify the probabilities with which the two outcomes occur:
\medskip

\begin{algorithm}[H]
	\DontPrintSemicolon 
	\NoCaptionOfAlgo
	\algotitle{\textsc{Main-XOS}}{fig:XOS-mechanism-main-r.title}
	With probability $p=0.08$, pick a most valuable item as the only winner and pay him $B$  \;
	With probability $1-p$, run \nameref{fig:XOS-random-sample-r} \;
	\caption{\textsc{Main-XOS}} \label{fig:XOS-mechanism-main-r} 
\end{algorithm}\medskip 

By following a similar but more careful analysis, we improve the approximation ratio by a factor of 3 while also retaining its properties. Notice that like the mechanism of \cite{BeiCGL12}, \nameref{fig:XOS-mechanism-main-r} is randomized and has superpolynomial running time. In particular, it requires a demand oracle.

\begin{theorem}\label{thm:xos-m-m-r}
	\nameref{fig:XOS-mechanism-main-r} is  universally truthful, individually rational, budget-feasible,  and has approximation ratio $244$.
\end{theorem}


\begin{proof} 
	Universal truthfulness, individual rationality and budget feasibility follow directly from the proof of Theorem 3.1 of \cite{BeiCGL12}.
	We next prove the approximation ratio of the mechanism. Let  $i^*\in\argmax_{i\in A}v(i)$ and $X_S$ be the output of \nameref{fig:XOS-random-sample-r}. In addition, let $X$ be an optimal solution, i.e., a subset of $A$ such that $v(X)=\opt(A, v, \mathbf{c}, B)$ and $\mathbf c(X)\le B$.
	We  need the following two lemmata:
	
	\begin{lemma}[Claim 3.1 in  \cite{BeiCGL12}]\label{lem:3.1of6} 
		For any $S \subseteq S^*, \alpha(S)-t\cdot \mathbf c(S) \geq 0$.
	\end{lemma}
	
	\begin{lemma}[Lemma 2.1 in \cite{BeiCGL12}]\label{lem:2.1of6} 
		Consider any subadditive function $v(\cdot)$.  For any given subset $S \subseteq A$ and a positive integer $k$ assume that $v(S) \geq k \cdot v(i)$ for any $i \in S$. Further, suppose that $S$ is divided uniformly at random into two groups $T_1$ and $T_2$. Then with probability at least $\frac{1}{2}$, we have that $v(T_1) \geq \frac{k-1}{4k}v(S)$ and $v(T_2) \geq \frac{k-1}{4k}v(S).$
	\end{lemma}
	
	Let $\kappa=74$.  We are going to use $\kappa$ to set different values for $k$ in Lemma \ref{lem:2.1of6} for different cases. 
	We follow a similar analysis as \cite{BeiCGL12},  but since we use two different values for $k$, things get a little more complicated. 
	Let $X_M$ to be the output of \nameref{fig:XOS-mechanism-main-r}.  We have the following three case regarding the value of item $i^*$:
	\begin{enumerate}[leftmargin=*]
		\item $v(i^*)>\frac{3.8}{\kappa}\opt(A, v, \mathbf{c}, B)$. In this case we have that 
		\[\mathrm E(v(X_M))=p\cdot v(i^*)+ (1-p)\cdot \mathrm E(v(X_S)) \ge p \cdot v(i^*)\ge p  \frac{3.8}{\kappa} \opt(A, v, \mathbf{c}, B)\,,\] and thus, $\opt(A, v, \mathbf{c}, B) \leq 243.5\cdot \mathrm E(v(X_M))$. \medskip
		
		\item $ \frac{1}{\kappa}\opt(A, v, \mathbf{c}, B)<v(i^*)\leq \frac{3.8}{\kappa}\opt(A, v, \mathbf{c}, B)$. Note that now we can apply Lemma \ref{lem:2.1of6} with $k=\frac{\kappa}{3.8}$. We split this case into two subcases:
		
		\begin{itemize} 
			\item $\mathbf c(S^*)>B$. 
			Since $\mathbf c(S^*)$ is more than $B$ we can find a subset $S' \subseteq S^*$, such that $\frac{B}{2} \leq \mathbf c(S') \leq B$. By  Lemma \ref{lem:3.1of6} we have 
			\[\alpha(S') \geq t \cdot \mathbf c(S') \geq \frac{\opt(T, v, \mathbf{c}_T, B)}{4.6\cdot B}\cdot\frac{B}{2} \geq \frac{\opt(T, v, \mathbf{c}_T, B)}{9.2}\,.\] 
			By using Lemma \ref{lem:2.1of6} we have
			\begin{IEEEeqnarray*}{rCl}
				\opt(T, v, \mathbf{c}_T, B)  & \geq & \opt(T\cap X, v, \mathbf{c}_{T\cap X}, B) \geq {\frac{\kappa-3.8}{4\kappa}}\opt(X, v, \mathbf{c}_X, B)   \\
				& = & {\frac{\kappa-3.8}{4\kappa}}\opt(A, v, \mathbf{c}, B)\,. 
			\end{IEEEeqnarray*} 
			Since $\opt(S^*, \alpha, \mathbf{c}_{S^*}, B)$ is the value of an optimal solution and $S'$ a particular solution under budget constraint $B$, we  conclude that 
			\[\opt(S^*, \alpha, \mathbf{c}_{S^*}, B) \geq \alpha(S') \ge \frac{\opt(T, v, \mathbf{c}_{T}, B)}{9.2} \geq {\frac{(\kappa-3.8) }{36.8\cdot\kappa}} \opt(A, v, \mathbf{c}, B)\,,\] 
			with probability at least $\frac{1}{2}$.\smallskip
			
			\item $\mathbf c(S^*)\leq B$. 
			In this case we have that \[\alpha(S^*)= \opt(S^*, \alpha, \mathbf{c}_{S^*}, B) = \opt(S^*, v, \mathbf{c}_{S^*}, B)=v(S^*)\,.\] Let $S'=X \mysetminus\, T$. We have $\mathbf c(S') \leq \mathbf c(X) \leq B$ and also, by using  Lemma \ref{lem:2.1of6}, $v(S') \geq \frac{\kappa-3.8}{4\kappa}\opt(A, v, \mathbf{c}, B)$ with probability of at least $\frac{1}{2}$. Then, we have  
			\begin{IEEEeqnarray*}{rCl}
				\opt(S^*, \alpha, \mathbf{c}_{S^*}, B) & = & v(S^*)  \geq  v(S^*)-t \cdot \mathbf c(S^*) \geq  v(S')-t \cdot \mathbf c(S')   \\
				& \geq & \frac{\kappa-3.8}{4\kappa}\opt(A, v, \mathbf{c}, B)- \frac{\opt(T, v, \mathbf{c}_T, B)}{4.6 \cdot B}\cdot B\\
				& \geq & \left( \frac{\kappa-3.8}{4\kappa}- \frac{1}{4.6}\right)  \opt(A, v, \mathbf{c}, B) \,,
			\end{IEEEeqnarray*}
			with probability at least $\frac{1}{2}$.
		\end{itemize}
		
		By substituting $\kappa$, it is easy to see that $\frac{\kappa-3.8}{4\kappa}- \frac{1}{4.6} < \frac{(\kappa-3.8) }{36.8\cdot\kappa}$. So, in both cases we have that $\opt(S^*, \alpha, \mathbf{c}_{S^*}, \allowbreak B)  \geq \left( \frac{\kappa-3.8}{4\kappa}- \frac{1}{4.6}\right) \opt(A, v, \mathbf{c}, B)$ with probability at least $\frac{1}{2}$. Now recall that \textsc{Additive-Mechanism} has an approximation factor of at most 3 with respect to $\opt(S^*, \alpha, \mathbf{c}_{S^*}, B)$. So we can finally derive that 
		\[\mathrm E(v(X_S)) \geq \frac{1}{3} \opt(S^*, \alpha, \mathbf{c}_{S^*}, B) \geq \frac{1}{3}\cdot \frac{1}{2}\cdot \left( \frac{\kappa-3.8}{4\kappa}- \frac{1}{4.6}\right)\cdot \opt(A, v, \mathbf{c}, B)\,.\] 
		Thus the solution that \nameref{fig:XOS-mechanism-main-r} returns has expected value 
		\begin{IEEEeqnarray*}{rCl}
			\mathrm E(v(X_M)) & = & p\cdot v(i^*)+(1-p)\cdot \mathrm E(v(X_S))   \\
			& \geq & \frac{p}{\kappa}\cdot \opt(A, v, \mathbf{c}, B) + \frac{(1-p)}{6} \cdot \left( \frac{\kappa-3.8}{4\kappa}- \frac{1}{4.6}\right) \cdot \opt(A, v, \mathbf{c}, B) \,.
		\end{IEEEeqnarray*}
		By substituting the values for $p, \kappa$ we  get $\opt(A, v, \mathbf{c}, B) \leq 243.2 \cdot \mathrm E(v(X_M))$.
		\medskip
		
		\item $v(i^*)\leq\frac{1}{\kappa}\opt(A, v, \mathbf{c}, B)$. The analysis of case 2 holds here as well, so we omit the details. The only difference is that
		now Lemma \ref{lem:2.1of6} should be applied with $k=\kappa$.
		The subcase where $\mathbf c(S^*)> B$ gives $\opt(S^*, \alpha, \mathbf{c}_{S^*}, B) \geq \frac{(\kappa-1) }{36.8\cdot\kappa}  \opt(A, v, \mathbf{c}, B)$ with probability at least $\frac{1}{2}$, 
		while the subcase where $\mathbf c(S^*)\leq B$ gives $\opt(S^*, \alpha, \mathbf{c}_{S^*}, B) \geq \left( \frac{\kappa-1}{4\kappa}- \frac{1}{4.6}\right)  \opt(A, v, \mathbf{c}, B)$ with probability at least $\frac{1}{2}$.
		By substituting $\kappa$, we see that $\frac{(\kappa-1) }{36.8\cdot\kappa} < \frac{\kappa-1}{4\kappa}- \frac{1}{4.6}$. So, in both cases we have that $\opt(S^*, \alpha, \mathbf{c}_{S^*}, B) \geq \frac{(\kappa-1) }{36.8\cdot\kappa} \opt(A, v, \mathbf{c}, B)$ with probability at least $\frac{1}{2}$. The above analysis gives that the solution returned by \nameref{fig:XOS-mechanism-main-r} has expected value 
		\[\mathrm E(v(X_M))\geq  \frac{(1-p)}{6} \cdot \frac{(\kappa-1) }{36.8\cdot\kappa} \cdot \opt(A, v, \mathbf{c}, B)\,.\] 
		By substituting the values for $p, \kappa$ we  get $\opt(A, v, \mathbf{c}, B) \leq 243.3 \cdot \mathrm E(v(X_M))$.
	\end{enumerate}
	We conclude that  $\opt(A, v, \mathbf{c}, B) \leq 244 \cdot E(v(X_M))$.
\end{proof}

\section{Conclusions}
\label{sec:conclusions}

In this work, we have made progress on the design of budget-feasible mechanisms for symmetric submodular objectives. The highlights of our results are polynomial time algorithms for the Budgeted Max Cut problem (weighted and unweighted) with significant improvements over previously known approximation factors. Although for general symmetric submodular functions we have exponential running times, the results imply polynomial time algorithms for any objective, where we can bound the optimal fractional solution with respect to the integral one. 
These results make further progress on the questions posed by \cite{DobzinskiPS11}. It remains an open problem however, whether any of these approximation ratios are tight.
Regarding our techniques, 
we expect that the idea of utilizing local search in order to identify monotone regions of a general submodular function may have further applications on mechanism design. Finally, apart from the mechanism design problem, the algorithmic result of Section \ref{sec:alg_ssm} is an interesting consequence for submodular optimization.

\newpage


{\small
\bibliographystyle{plain}
\bibliography{budgetedMechanismDesign}

}

\newpage

\appendix

\section{Missing Material from Section \ref{sec:prels}}
\label{app:prels}

%
%

\subsection{Instances with Costs Exceeding the Budget}
\label{app:costs}
Consider an instance with a symmetric submodular function, where there exist agents with $c_i > B$. It may seem at first sight that we could just discard such agents, since too expensive agents are not included in any feasible solution anyway. However, the presence of such agents can create  infeasible solutions of very high value and make an analog of Lemma \ref{lem:almost-monotone} impossible to prove. Simply discarding them could also destroy the symmetry of the function (e.g., if we had a cut function defined on a graph, we could not just remove a node). 

Let $\mathcal{I}$ denote the set of all instances of the problem with symmetric submodular functions, and let $\mathcal{J}$ denote the set of all instances where at most one agent has cost more than $B$. Given $X\subseteq A$, we let $\mathbf{c}(X) = \sum_{i\in X} c_i$.
The next lemma, together with its corollary, show that, when dealing with symmetric submodular functions, we may only consider instances in $\mathcal{J}$ without any loss of generality. 

\begin{lemma}\label{lem:Costs}
Given an instance $I=(A, v, \mathbf{c}, B)\in \mathcal{I}$, we can efficiently construct an instance $J=(A', v', \mathbf{c}', B)\in \mathcal{J}$ such that
\begin{itemize}
\item Every feasible solution of $I$ is a feasible solution of $J$ and vice versa. 
\item If $X$ is a feasible solution of $I$, then $v(X)=v'(X)$ and $\mathbf{c}(X) = \mathbf{c'}(X)$. In particular, $\opt(J)=\opt(I)$.
\end{itemize}
\end{lemma}

\begin{proof}
Let $E=\{i\in A\ |\ c_i>B\}$ be the set of expensive agents and define $A'=(A\mysetminus E)\cup \{i_E\}$, where $i_E$ is a new agent replacing the whole set $E$. For $i\in A\mysetminus E$ we define $c'_i=c_i$, while $c'_{i_E}=B+1$. Finally, $v'$ is defined as follows
\begin{displaymath}
v'(T)= \left\{ \begin{array}{ll}
v(T)\,, & \text{if\ \  } T\subseteq A\mysetminus E\\
v((T\mysetminus \{i_E\}) \cup E)\,, & \text{otherwise}
\end{array} \right.
\end{displaymath}

Now suppose $X$ is a budget-feasible solution of $I$. Then $\mathbf{c}(X)\le B$ and thus $X\subseteq A\mysetminus E$. But then, by the definition of $\mathbf{c}'$, $\mathbf{c}'(X)=\mathbf{c}(X)\le B$ as well, and therefore $X$ is also a budget-feasible solution of $J$. Moreover, $v'(X)=v(X)$ by the definition of $v'$. We conclude that $\opt(I)\le\opt(J)$.

The proof that  every feasible solution of $J$ is a feasible solution of $I$ is almost identical. This implies $\opt(J)\le\opt(I)$, and therefore $\opt(J)=\opt(I)$.
\end{proof}

Now, it is not hard to see that we can turn any algorithmic result on $\mathcal{J}$ to the same algorithmic result on $\mathcal{I}$. However, we need a somewhat stronger statement to take care of issues like truthfulness and budget-feasibility. This is summarized in the following corollary.

\begin{corollary}\label{cor:Costs}
Given a (polynomial time) algorithm $\textsc{alg}'$ that achieves a $\rho$-approximation on instances in $\mathcal{J}$, we can efficiently construct a (polynomial time) $\rho$-approximation algorithm $\textsc{alg}$ that works for all instances in $\mathcal{I}$. Moreover, if $\textsc{alg}'$ is monotone and  budget-feasible on instances in $\mathcal{J}$, assuming Myerson's threshold payments, then $\textsc{alg}$ is monotone and budget-feasible on instances in $\mathcal{I}$.
\end{corollary}

\begin{proof}
%

The description of $\textsc{alg}$ is quite straightforward. Given an instance $I=(A, v, \mathbf{c}, B)\in \mathcal{I}$, $\textsc{alg}$ first constructs instance $J=(A', v', \mathbf{c}', B)\in \mathcal{J}$, as described in the proof of Lemma \ref{lem:Costs}. Then $\textsc{alg}$ runs  $\textsc{alg}'$ with input $J$ and returns its output. Clearly, if $\textsc{alg}'$ runs in polynomial time, so does $\textsc{alg}$.

If $X=\textsc{alg}'(J)=\textsc{alg}(I)$, then $X$ is feasible with respect to $J$ and $\opt(J)\le \rho\cdot v'(X)$. By Lemma \ref{lem:Costs} we get that $X$ is feasible with respect to $I$ and $\opt(I) = \opt(J) \le \rho\cdot v'(X) = \rho\cdot v(X)$. This establishes the approximation ratio of $\textsc{alg}$. 

Next, assume that $\textsc{alg}'$ is monotone and---assuming Myerson's threshold payments---budget-feasible on instances in $\mathcal{J}$. Suppose that agent $j\in \textsc{alg}(I)$ reduces his cost from $c_j$ to $b_j < c_j$. This results in a new instance $I_*=(A, v, (b_j, \mathbf{c}_{-j}), B)\in \mathcal{I}$. Since it must be the case where $c_j\le B$, the corresponding instance of $\mathcal{J}$ is $J_*=(A', v', (b_j, \mathbf{c}'_{-j}), B)$. 
Due to the monotonicity of $\textsc{alg}'$ we have 
\[j\in \textsc{alg}(I) = \textsc{alg}'(J) \Rightarrow  j\in \textsc{alg}'(J_*) = \textsc{alg}(I_*) \,, \]
and therefore $\textsc{alg}$ is monotone as well. 

The budget-feasibility of $\textsc{alg}$ follows from the budget-feasibility of $\textsc{alg}'$ by observing that $i\in \textsc{alg}(A, v,\allowbreak (b_j, \mathbf{c}_{-j}), B)$ if and only if $i\in \textsc{alg}'(A', v', (b_j, \mathbf{c}_{-j}), B)$ 
for all $i\in A'$.
\end{proof}

\subsection{Regarding Remark \ref{rem:n-sub}}\label{app:counterexample}

We use the cut function on a very simple graph and show that although $v$ is submodular, $\hat{v}$ is not. Consider the following graph where each edge has unit weight:
\begin{center}
{\scalebox{0.255} {\includegraphics{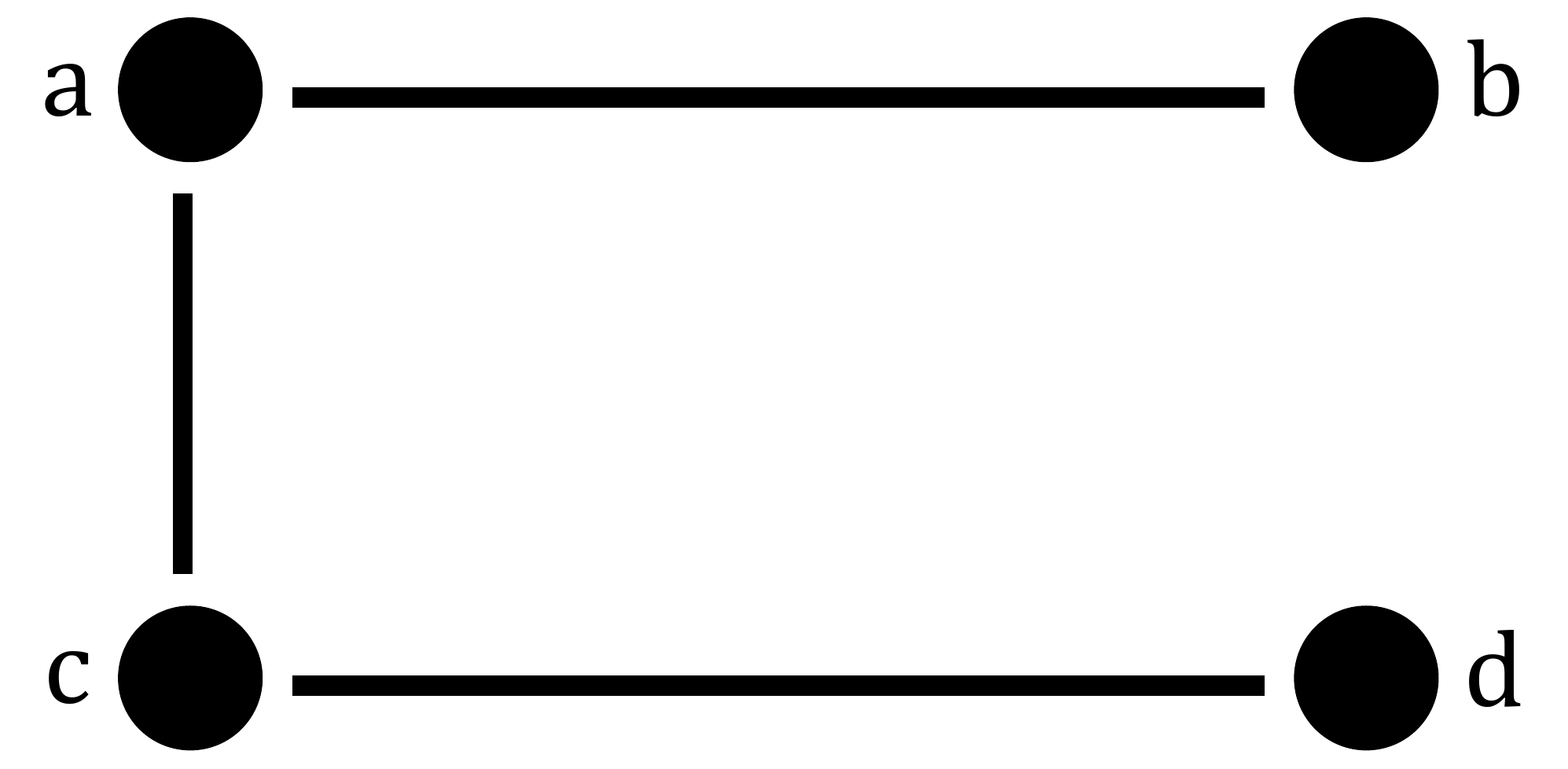}}}
\end{center}
We compute the value of the following sets:
\begin{itemize}
\item[ ] $\hat{v}(\{a\})=v(\{a\})=2$
\item[ ] $\hat{v}(\{a, b\})=v(\{a\})=2$
\item[ ] $\hat{v}(\{a, c\})=v(\{a\})=v(\{c\})=v(\{a, c\})=2$
\item[ ] $\hat{v}(\{a, b, c\})=v(\{b, c\})=3$
\end{itemize}
Now it is easy to see that although $\{a\} \subseteq \{a, c\}$ we have  
$\hat{v}(\{a\}\cup \{b\})-\hat{v}(\{a\})=0 < 1= \hat{v}(\{a, c\}\cup\{b\})-\hat{v}(\{a, c\})$.

So an interesting question is if $\hat{v}$ can be classified when $v$ is submodular. In \cite{GuptaNS17} the general XOS class was introduced, where it is allowed for a function to be non-monotone (recall here that the XOS class contains only non-decreasing functions  by definition). They proved that when $v$ is  general XOS then $\hat{v}$ is XOS, while in addition they observed that the class of non-negative submodular functions is a strict subset of the general XOS class. Thus since any non-negative submodular function $v$ is also general XOS we conclude $\hat{v}$ is XOS. 

\section{Proof of Theorem \ref{thm:alg_ssm}}
\label{app:alg_ssm}
To facilitate the exposition of the analysis, we restate algorithm \nameref{fig:alg-greedy-enum}, so that instead of the set variable $T$, we have a variable $S^t$ that describes the set constructed at iteration $t$. 
\smallskip

\begin{algorithm}[H]
	\DontPrintSemicolon 
	\NoCaptionOfAlgo
	\algotitle{\textsc{Var-Greedy-Enum-SM}}{fig:var-alg-greedy-enum.title}
	Let $S_1$ be the best feasible solution of cardinality at most 3 (by enumerating them all) \;
	$S_2=\emptyset$ \;
	\For{every $U\subseteq A$ with $|U|=3$}{
		$S^0=U,\ t=1,\ A^0=A\mysetminus U$ \;
		\While{$A^{t-1} \neq \emptyset$}{
			Find $\theta_t = \displaystyle{\max_{i\in A^{t-1}} \frac{v(S^{t-1}\cup\{i\})-v(S^{t-1})}{c_i}}$,  and let $i_t$ be an  element of $A^{t-1}$ that attains $\theta_t$ \;
			\If{$\theta_t\ge 0$ and $\sum_{i\in S^{t-1}\cup\{i_t\}}c_i \le B$}{$S^t = S^{t-1}\cup\{i_t\}$}
			\Else{$S^t = S^{t-1}$}
			$A^t = A^{t-1}\mysetminus\{i_t\}$ \;
			$t=t+1$ \;
		}
		\If{$v(S^{t-1})>v(S_2)$}{$S_2=S^{t-1}$}
	}
	\Return $S \in \displaystyle\argmax_{X\in \{S_1, S_2\}} v(X)$
	\caption{\textsc{Var-Greedy-Enum-SM}$(A, v, \mathbf{c}, B)$} \label{fig:var-alg-greedy-enum} 
\end{algorithm}\smallskip 

Recall that \nameref{fig:var-alg-greedy-enum} runs on both $S$ and $A\mysetminus S$ and \nameref{fig:alg-lsgreedy} returns the best solution of these two.
We may assume, without loss of generality that $\opt(S, B)= \max\{\opt(S, B), \opt(A\mysetminus S, B)\}$ (the case for $A\mysetminus S$ being symmetric).
By Fact \ref{fact:simple-facts} we have $\opt(S, B) \ge 0.5 \opt(A, B)$.
So, it suffices to show  that running \nameref{fig:var-alg-greedy-enum} on $S$ outputs a set of value at least $(1-1/e-\epsilon)\opt(S, B)$.

In what follows we analyze the approximation ratio achieved by \nameref{fig:var-alg-greedy-enum}($S, v, \mathbf{c}, B$) with respect to $\opt(S, B)$. For this, we follow closely the proof of the main result in \cite{Sviridenko04}.

If there is an optimal solution for our problem restricted on $S$, of cardinality one, two or three, then the set $S_1$ of \nameref{fig:var-alg-greedy-enum.title} will be such a solution. Hence, assume that
the cardinality of any optimal solution is at least four and let $S^*$ be such a solution.
If necessary, reorder the elements of $S^* = \{j_1,\ldots, j_{|S^*|}\}$ so that 
$j_{1} = \arg\max_{\ell} v(\{j_{\ell}\})$, and
$j_{k+1} = \arg\max_{\ell > k} \left[ v(\{j_1, \ldots, j_k, j_{\ell}\}) - v(\{j_1, \ldots, j_k\}) \right]$. 

Let $Y=\{j_1, j_2, j_3\}$. For notational convenience, we will use the function $g(\cdot) = v(\cdot) -v(Y)$. It is straightforward that $g(\cdot)$ is submodular. Moreover, the following fact follows from \cite{Sviridenko04}.
\begin{fact}\label{fact:g}
	$g(X\cup \{i\}) - g(X) \le \frac{1}{3} v(Y)$, for any $Y \subseteq X \subseteq S$ and $i\in S^*\mysetminus X$.
\end{fact} 

Consider the execution of the greedy algorithm with initial set $U=Y$. 
Let $t^*+1$ be the first time when an element $i_{t^* +1} \in S^*$ is not added to $S^{t^*}$. 
In fact,  we assume that $t^* + 1$ is the first time when $S^t = S^{t-1}$.
(To see that this is without loss of generality, if there is some time $\tau < t^* + 1$ such that $i_{\tau}$ is not added to $S^{\tau -1}$, then---by the definition of $t^*+1$---it must be the case that $i_{\tau} \notin S^*$. But then, we may consider the instance $(S\mysetminus\{i_\tau\}, v, \mathbf{c}_{S\mysetminus\{i_\tau\}}, B)$ instead. We have $v(S^*) = \opt(S\mysetminus\{i_\tau\}, B)=\opt(S, B)$ and the greedy solution constructed in the iteration where $S_0 = Y$ is exactly the same as before.)
We are going to distinguish two cases.
\smallskip

\noindent \textbf{Case 1.} For all $t\in [t^*], \theta_t \ge 0$, but $\theta_{t^*+1} < 0$. Using Theorem \ref{thm:alt_SM} for $S^*$ and  $S^{t^*}$ we have 
\begin{IEEEeqnarray*}{rCl}
g(S^*) & \le & g(S^{t^*}) + \sum_{i \in S^* \mysetminus S^{t^*}} (g(S^{t^*}\cup\{i\}) - g(S^{t^*})) - \sum_{i \in S^{t^*} \mysetminus S^*} (g(S^{t^*}\cup S^*) - g(S^{t^*}\cup S^* \mysetminus \{i\})) \\
 & = & g(S^{t^*}) + \sum_{i \in S^* \mysetminus S^{t^*}} (v(S^{t^*}\cup\{i\}) - v(S^{t^*})) - \sum_{i \in S^{t^*} \mysetminus S^*} (v(S^{t^*}\cup S^*) - v(S^{t^*}\cup S^* \mysetminus \{i\})) \\
 & \le & g(S^{t^*}) + \sum_{i \in S^* \mysetminus S^{t^*}} c_i \theta_{{t^*}+1} - |S^{t^*} \mysetminus S^*| \left( - \frac{\epsilon}{n} \opt(S, B)\right) \\
 & \le & g(S^{t^*}) +  \epsilon \opt(S, B) \,,
\end{IEEEeqnarray*}
Here, the second to last inequality holds by Lemma \ref{lem:almost-monotone} and by the assumptions we have made. That is, for every 
$i\in S^* \mysetminus S^{t^*}$, we have that $i\in A^{t^*}$, since we assumed that $t^*+1$ is the first time when $S^t = S^{t-1}$, hence up until time $t^*$, $A^{t^*}$ contains all the agents apart from $S^{t^*}$. This implies that for every $i\in S^* \mysetminus S^{t^*}$, we have that 
$v(S^{t^*}\cup\{i\}) - v(S^{t^*}) \leq c_i\theta_{t^*+1}$, by the definition of $\theta_{t^*+1}$.

Therefore, we can conclude that
\[v(S^{t^*})  =  v(Y)+ g(S^{t^*}) \ge v(Y)+ g(S^{*}) - \epsilon \opt(S, B) = (1-\epsilon)\opt(S, B) \,.\]

\noindent \textbf{Case 2.} For all $t\in [t^*+1], \theta_t \ge 0$, but $\sum_{i\in S^{t^*}\cup\{i_{t^*+1}\}}c_i > B$ while $\sum_{i\in S^{t^*}}c_i \le B$.
Using Theorem \ref{thm:alt_SM} for $S^*$ and each of $S^t$, $t\in [t^*]$, as well as Lemma \ref{lem:almost-monotone}, we have 
\begin{IEEEeqnarray*}{rCl}
g(S^*) & \le & g(S^t) + \sum_{i \in S^* \mysetminus S^t} (g(S^t\cup\{i\}) - g(S^t)) - \sum_{i \in S^t \mysetminus S^*} (g(S^t\cup S^*) - g(S^t\cup S^* \mysetminus \{i\})) \\
 & = & g(S^t) + \sum_{i \in S^* \mysetminus S^t} (v(S^t\cup\{i\}) - v(S^t)) - \sum_{i \in S^t \mysetminus S^*} (v(S^t\cup S^*) - v(S^t\cup S^* \mysetminus \{i\})) \\
 & \le & g(S^t) + \sum_{i \in S^* \mysetminus S^t} (v(S^t\cup\{i\}) - v(S^t)) - |S^t \mysetminus S^*| \left( - \frac{\epsilon}{n} \opt(S, B)\right) \\
 & \le & g(S^t) + \sum_{i \in S^* \mysetminus S^t} (v(S^t\cup\{i\}) - v(S^t)) + \epsilon \opt(S, B) \,,
\end{IEEEeqnarray*}
and therefore
\begin{IEEEeqnarray*}{rCl}
g(S^*) -\epsilon \opt(S, B) & \le & g(S^t) + \sum_{i \in S^* \mysetminus S^t} (v(S^t\cup\{i\}) - v(S^t))  \\
 & \le & g(S^t) + \sum_{i \in S^* \mysetminus S^t} c_i \theta_{t+1}  \\
 & \le & g(S^t) + \bigg( B - \sum_{i\in Y} c_i \bigg)\theta_{t+1} \,,
\end{IEEEeqnarray*}
for all $t\in [t^*]$.

For the last part of the proof we need the following inequality of Wolsey \cite{Wolsey82}.

\begin{theorem}[Wolsey \cite{Wolsey82}] \label{thm:ineq}
Let $k$ and $s$ be arbitrary positive integers, and $\rho_1, \ldots, \rho_k$ be arbitrary reals with $z_1= \sum_{i=1}^k \rho_i$ and $z_2=\min_{t\in [k]} \big( \sum_{i=1}^{t-1} \rho_i + s \rho_t \big)   >0$. Then $z_1 / z_2 \ge 1 - (1-1/s)^k \ge 1- e^{-\frac{k}{s}}$. 
\end{theorem}

For any $\tau$, define $B_{\tau}=\sum_{t=1}^{\tau} c_{i_t}$, and let $k= B_{t^*+1}$ and $s= B - \sum_{i\in Y} c_i $. We also define $\rho_1, \ldots, \rho_k$ as follows; for $i\le c_1, \rho_i=\theta_1$ and for $B_{\tau} < i \le B_{\tau+1}, \rho_i=\theta_{\tau+1}$. It is easy to see that $g(S^{t^*}\cup\{i_{t^*+1}\}) = \sum_{t=1}^{t^*+1} c_{i_t} \theta_{i_t} = \sum_{i=1}^k \rho_i$ and similarly $g(S^{t}) = \sum_{t=1}^{\tau} c_{i_t}\theta_{i_t} = \sum_{i=1}^{B_{\tau}} \rho_i$. 

As noted in \cite{Sviridenko04}, since the $\rho_i$s are nonnegative we have
\[\min_{t\in [k]} \left(  \sum_{i=1}^{t-1} \rho_i + s \rho_t \right)  = \min_{\tau\in [t^*]} \left(  \sum_{i=1}^{B_{\tau}} \rho_i + s \rho_{B_{\tau}+1} \right) \,,\]
and therefore
\[g(S^*) -\epsilon \opt(S, B)\le \min_{t\in [k]} \left(  \sum_{i=1}^{t-1} \rho_i + s \rho_t \right) \,.\]
So, as a direct application of Theorem \ref{thm:ineq} we have
\[\frac{g(S^{t^*}\cup\{i_{t^*+1}\})}{g(S^*) -\epsilon \opt(S, B)} \ge 1- e^{-\frac{k}{s}} > 1- e^{-1} \,.\]
Finally, using the above inequality and Fact \ref{fact:g}, we get
\begin{IEEEeqnarray*}{rCl}
v(S^{t^*}) & = & v(Y)+ g(S^{t^*}) =  v(Y)+ g(S^{t^*}\cup\{i_{t^*+1}\})- (g(S^{t^*}\cup\{i_{t^*+1}\}) - g(S^{t^*}))  \\
 & \ge & v(Y) + (1- e^{-1}) g(S^*) - (1- e^{-1}) \epsilon \opt(S, B) - (v(S^{t^*}\cup\{i_{t^*+1}\}) - v(S^{t^*})) \\
 & \ge & v(Y) + (1- e^{-1}) g(S^*) - \epsilon \opt(S, B) - \frac{1}{3}v(Y) \\
 & \ge & (1- e^{-1}- \epsilon) \opt(S, B) \,.
\end{IEEEeqnarray*}
 
Since in both cases the final output $T^*$ of the algorithm has value at least $v(S^{t^*})$, this implies that
\[v(T^*) \ge  (1- e^{-1}- \epsilon) \opt(S, B) \ge \frac{1- e^{-1}- \epsilon}{2} \opt(A, B) > \Big(\frac{e-1}{2e} - \epsilon \Big) \opt(A, B) \,,\]
thus concluding the analysis of the performance of the algorithm.
\qed

\section{Missing Proofs from Section \ref{sec:exp_ssm}}
\label{app:bumcut}

\begin{proofof}{Proof of Theorem \ref{thm:mechSM-frac}}
Monotonicity (and therefore truthfulness and individual rationality) and budget-feasibility of the  mechanism directly follow from \cite{ChenGL11} and \cite{ABM16}. Next we show  the approximation guarantee.
As with the proof of Theorem \ref{thm:mechSM}, Lemma \ref{lem:greedy_approximation_jt} is crucial.

Let $S$  denote the outcome of \nameref{fig:alg-2}$(A, v, \mathbf{c}, B/2)$.
We also consider two cases. Let $\alpha = \rho +1 +\sqrt{\rho^2 +4\rho +1}$.
\begin{itemize}
	\item If $i^*$ is returned, then 
	$\alpha \cdot v(i^*) \ge \opt_f(A'\mysetminus\{i^*\},   B) \ge \opt(A'\mysetminus\{i^*\},   B) = \opt(A\mysetminus\{i^*\},   B) \ge \opt(A,   B) - v(i^*)$, 
	and therefore $\opt(A,   B)\le (\alpha +1) \cdot v(i^*)$. 
	\item On the other hand, if  $S$ is returned, then
	$\alpha \cdot v(i^*) < \opt_f(A'\mysetminus\{i^*\},   B) \le \rho\cdot \opt(A'\mysetminus\{i^*\},   B) \le \rho\cdot \opt(A,   B)$. 
	Combining this with Lemma \ref{lem:greedy_approximation_jt} we have 
	$\opt(A,  B) \le 3\cdot v(S)+\frac{2\rho}{\alpha}\opt(A,   B)$ and therefore $\opt(A,  B) \le \frac{3\alpha}{\alpha -2\rho}\cdot v(S) = (\alpha +1)  \cdot v(S)$, where the last equality is only a matter of calculations.\qed
\end{itemize}
\end{proofof}
\bigskip

\begin{proofof}{Proof of Theorem \ref{lem:bwmcut}}
The  proof takes the same approach as the proof of Lemma 3 from \cite{ABM16}. 
We begin with a nonlinear program such that if all the $x_i$s are integral then the objectives \eqref{eq:bwmcut1} and \eqref{eq:nlbwmcut1} have the same value.
\begin{IEEEeqnarray}{rCl}
\mbox{maximize:\ \ } & & F(x) = \sum_{i\in [n]}\sum_{j\in [n]\mysetminus[i]} w_{ij} (x_i +x_j - 2x_i x_j)   \label{eq:nlbwmcut1}\\
\mbox{subject to:\ \ }& & \sum_{i\in [n]} c_i x_i \le B \label{eq:nlbwmcut2} \\
& & x_i = 0\ ,\quad  \forall i\in [n]\mysetminus X\\
& & 0\le x_i \le 1\ , \quad \forall i\in [n] \label{eq:nlbwmcut3}
\end{IEEEeqnarray}
So, if $x$ is any feasible integral vector to our problem, we have $F(x)\le \opt(I)$. Moreover, given any feasible solution $x, z$ to \eqref{eq:bwmcut1}-\eqref{eq:bwmcut5}, the value of \eqref{eq:bwmcut1} is upper bounded by $L(x)=  \sum_{i\in [n]}\sum_{j\in [n]\mysetminus[i]} w_{ij}  \min\{x_i +x_j, 2-x_i -x_j \}$, since $z_{ij} \le \min\{x_i +x_j, 2-x_i -x_j \}$ for any $i\in [n], j \in [n] \mysetminus [i]$. 

Next, we show that $F(x)\ge 0.5 L(x)$. This follows from the inequality 
\[2(a +b - 2ab) \ge \min\{ a+b, 2-a -b \} \text{ for any }a, b\in [0,1]\,, \]
proven by Ageev and Sviridenko \cite{AgeevS99}. For completeness we prove it here as well. Notice that by replacing $a$ and $b$ by $1-a$ and $1-b$ respectively both sides of the inequality remain exactly the same. Therefore, it suffices to prove 
$2(a +b - 2ab) \ge a + b $ for any $a, b\in [0,1]$ such that $a+b \le 1$ (since otherwise, $1-a + 1-b \leq 1$). This however is equivalent to $a+b\ge 4ab$ which is true since $a+b\ge (a+b)^2 \ge 4ab$.

Hence, if $x^*, z^*$ is an optimal fractional solution to \eqref{eq:bwmcut1}-\eqref{eq:bwmcut5}, then the value of \eqref{eq:bwmcut1} is 
$\opt_f(I) = L(x^*)$ and thus $F(x^*)\ge 0.5L(x^*) = 0.5 \opt_f(I)$.
However, $x^*$ may have several fractional coordinates. Our next step is to transform $x^*$ to a vector $x'$ that has at most one fractional coordinate and at the same time $F(x')\ge F(x^*)$. To this end, we show how to reduce the fractional coordinates by (at least) one in any feasible vector with at least two such coordinates.

Consider a feasible vector $x$, and suppose $x_i$ and $x_j$ are two non integral coordinates. Note that $i, j\in X$. Let $x^{i,j}_{\epsilon}$ be the vector we get if we replace $x_i$ by $x_i+\epsilon$ and $x_j$ by  $x_j - \epsilon c_i / c_j$ and leave every other coordinate of $x$ the same. Note that the function $\bar{F}(\epsilon)=F(x^{i,j}_{\epsilon})$, with respect to $\epsilon$, is either linear or a polynomial of degree 2 with positive leading coefficient. That is, $\bar{F}(\epsilon)$ is convex. 
	
Notice now that $x^{i,j}_{\epsilon}$ always satisfies the budget constraint \eqref{eq:nlbwmcut2}, and also satisfies \eqref{eq:nlbwmcut3} as long as 
$\epsilon\in \left[ \max\left\lbrace -x_i, (x_j - 1)c_j / c_i\right\rbrace  , \min\left\lbrace 1-x_i, x_j c_j / c_i\right\rbrace \right]$.
Due to convexity, $\bar{F}(\epsilon)$ attains a maximum on one of the endpoints of this interval, say at $\epsilon^*$. Moreover, at either endpoint at least one of $x_i+\epsilon^*$ and $x_j - \epsilon^* c_i/c_j$ is integral. That is, $x^{i,j}_{\epsilon^*}$ has at least one more integral coordinate than $x$ and $F(x^{i,j}_{\epsilon^*})\ge F(x)$.
%
%

So, initially $x\leftarrow x^*$. As long as there exist two non integral coordinates $x_i$ and $x_j$ we set $x\leftarrow x^{i,j}_{\epsilon^*}$ as described above. This happens at most $n-1$ times, and outputs a feasible vector $x'$ with at most one non-integral coordinate, and with $F(x')\ge F(x^*)$. We have then the following implications:
\begin{equation}
\opt_f(I)  =  L(x^*) \le 2 \cdot F(x^*) \le 2 \cdot F(x')\,. \label{eq:cut-ineq1}
\end{equation}

\noindent If $x'$ is integral, then by \eqref{eq:cut-ineq1} we have 
$\opt_f(I)\le 2 \cdot F(x') \le  2 \cdot \opt(I)$,
and we are done. So, suppose that $x_r'$ is the only fractional coordinate of $x'$. Let $x^{0}$ and $x^{1}$ be the vectors we get if we set $x_r'$ to $0$ or $1$ respectively and leave every other coordinate of $x'$ the same. 
Notice that for $a, b\in [0,1]$ the inequality $(1-a)(b-1)\le 0$ implies $a+b-ab\le 1$ and therefore $a+b-2ab\le 1$, so we have
%
\begin{IEEEeqnarray*}{rCl}
	F(x')-F(x^{0}) & = & \sum_{j\in [n]\mysetminus r}w_{r j} (x_r' - 2x_r'  x_j')  \\
	& \le & \sum_{j\in [n]\mysetminus r}w_{r j} (x_r' +x_j' - 2x_r'  x_j')  \\
	& \le & \sum_{j\in [n]\mysetminus r}w_{r j} = F(x^{1}-x^{0}) \,,
\end{IEEEeqnarray*}
and thus 
\begin{equation}
F(x') \le F(x^{0}) + F(x^{1}-x^{0})\,. \label{eq:cut-ineq2}
\end{equation}
Combining \eqref{eq:cut-ineq1}and \eqref{eq:cut-ineq2}, we have  
\[\opt_f(I) \le 2 \cdot F(x') \le  2 \cdot  \left( F(x^{0}) + F(x^{1}-x^{0}) \right) \,, \]

Using the fact that $F$ is upper bounded by $\opt(I)$ on integral vectors, we have that $F(x^{0}) \leq \opt(I)$. 
Observe now also that $F(x^{1}-x^{0}) = \sum_{j\in [n]\mysetminus r}w_{r j} = v(r) \leq \beta_I \opt(I)$, by the definition of $\beta_I$.
Hence, overall we get
\[ \opt_f(I) \le (2 + 2\beta_I) \opt(I)\,,\]
thus completing the proof.
\qed
\end{proofof}
\section{Missing Material from Section \ref{sec:poly_ssm}}
\label{app:bwmcut}

\subsection{Going Beyond Non-Decreasing Objectives}
\label{app:proof_of_greedy}
Here we  pave the way for the main results of Section \ref{sec:poly_ssm}, and therefore we need to prove that certain mechanisms work even for objectives that are not exactly non-decreasing. To make this precise, given a ground set $A$, a budget $B$ and a constant $\epsilon \ge 0$, we say that a set function $v$ is \textit{$(B, \epsilon)$-quasi-monotone} (or just quasi-monotone) on a set $X\subseteq A$ if for every $T\subsetneq X$ and every $i \in X\mysetminus T$, we have $v(T\cup\{i\}) - v(T) \ge - \frac{\epsilon}{n} \opt(X,  B)$. Clearly, $(B, 0)$-quasi-monotone on $X$ just means non-decreasing on $X$.

The main lemmas needed for our proofs are about \nameref{fig:alg-2}, as it all boils down to the monotonicity, budget-feasibility, and approximation ratio of this simple mechanism. As mentioned in Lemma \ref{lem:GreedySM}, \nameref{fig:alg-2} is monotone, since any item out of the winning set remains out of the winning set if it increases its cost.

\begin{lemma}\label{lem:greedy_budget-feasible}
Suppose $v$ is a $(B, \epsilon)$-quasi-monotone submodular function on $A$ such that $U\cdot v(S)\ge \opt(A,  B)$, where $S$ is the set output by \nameref{fig:alg-2}$(A, v, \mathbf{c}, (1-U \epsilon)B/2)$ and $U$ is a constant.\footnote{Typically, in our case, $U$ is a constant associated with the constant $\rho$ that determines how the optimal solution to the relaxed problem is bounded by the optimal solution to the original problem. In particular, throughout this work, $U$ is upper bounded by $(1+\rho)\big( 2+\rho + \sqrt{\rho^2 +4 \rho +1}\big)+1$.}
Assuming the payments of Myerson's lemma, $S$ is budget-feasible.
\end{lemma}

\begin{proof}
Recall that before the description of \nameref{fig:alg-2} we assumed---without loss of generality---that agents are sorted in descending order with respect to their ratio of marginal value over cost, i.e., $1=\argmax_{j\in A}\frac{v(j)}{c_j}$ and $i=\argmax_{j\in A\setminus [i-1]}\frac{v([j])-v([j-1])}{c_j}$ for $i\ge 2$. Suppose that $S=[\ell]$, i.e., $1, 2, \ldots, \ell$ are added in $S$ in that order.  Let $S_0=\emptyset = [0]$ and $S_i=[i]$ for $1\le i \le \ell$. We are going to show that the payment to agent $i$ is upper bounded by $\frac{B \cdot (v([i])-v([i-1])}{v(S)}$, and then budget feasibility directly follows from $\sum_{i\in[\ell]}\frac{B \cdot (v([i])-v([i-1])}{v(S)} = B$.

Suppose this upper bound does not hold for every agent. That is, there exists some $j\in[\ell]$ such that agent $j$ bids $b_j > \frac{B \cdot (v([j])-v([j-1]))}{v(S)}$ and is still included in the output $S'$ of \nameref{fig:alg-2}$(A, v, (\mathbf{c}_{-j}, b_j), (1-U\epsilon)B/2)$. Let $\mathbf{b} = (\mathbf{c}_{-j}, b_j)$ and notice that up to agent $j-1$, agents are added to $S'$ in the same order as they do in $S$, but after that the ordering might be affected. 
Also for $i\in[\ell]$ we have $\frac{v([i])-v([i-1])}{c_i} \ge  \frac{v([\ell])-v([\ell-1])}{c_\ell} \ge \frac{2\cdot v([\ell])}{(1-U\epsilon)B}$ and therefore
\[ \mathbf{b}(S\mysetminus \{j\}) = \mathbf{c}(S\mysetminus \{j\}) \le \mathbf{c}(S) = \sum_{i\in[\ell]}c_i \le \frac{(1-U\epsilon)B}{2\cdot v([\ell])} \sum_{i\in[\ell]}(v([i])-v([i-1]) =  \frac{(1-U\epsilon)B}{2} \,. \]

For \nameref{fig:alg-2}$(A, v, (\mathbf{c}_{-j}, b_j), (1-U\epsilon)B/2)$ let $S'_{j-1}$ be the chosen set right before $j$ is added and $S'_j = S'_{j-1}\cup \{j\}$. This implies 
\[j \in \argmax_{i\in[n]} \frac{v(S'_{j-1}\cup \{i\}) - v(S'_{j-1})}{b_i} \text{\ \ and\ \ } \frac{v(S'_{j}) - v(S'_{j-1})}{b_j}\ge \frac{2\cdot v(S'_{j})}{(1-U\epsilon)B} \,.\]
By Theorem \ref{thm:alt_SM} we have 
\begin{IEEEeqnarray*}{rCl}
v(S) - v(S'_{j}) & \le & \sum_{i \in S \mysetminus S'_{j}} (v(S'_{j}\cup\{i\}) - v(S'_{j})) - \sum_{i \in S'_{j} \mysetminus S} (v(S'_{j} \cup S) - v(S'_{j} \cup S \mysetminus \{i\})) \\
 & \le & \sum_{i \in S \mysetminus S'_{j}} (v(S'_{j}\cup\{i\}) - v(S'_{j})) - |S'_{j} \mysetminus S| \left( - \frac{\epsilon}{n} \opt(A,  B)\right) \\
 & \le & \sum_{i \in S \mysetminus S'_{j}} (v(S'_{j}\cup\{i\}) - v(S'_{j})) +  \epsilon U  v(S) \,.
\end{IEEEeqnarray*}
If $S \mysetminus S'_{j} = \emptyset$ then we directly get $v(S'_{j})\ge (1-\epsilon U) v(S)$, otherwise we have
\begin{IEEEeqnarray*}{rCl}
(1-\epsilon U) v(S) - v(S'_{j}) & \le & \sum_{i \in S \mysetminus S'_{j}} b_i\cdot \frac{v(S'_{j}\cup\{i\}) - v(S'_{j})}{b_i} \\
 & \le & \max_{i \in S \mysetminus S'_{j}}\frac{v(S'_{j}\cup\{i\}) - v(S'_{j})}{b_i}\cdot \sum_{i \in S \mysetminus S'_{j}} b_i  \\
 & \le & \max_{i \in [n]}\frac{v(S'_{j}\cup\{i\}) - v(S'_{j})}{b_i}\cdot \mathbf{b}(S \mysetminus S'_{j})  \\
 & \le & \max_{i \in [n]}\frac{v(S'_{j-1}\cup\{i\}) - v(S'_{j-1})}{b_i}\cdot \mathbf{b}(S\mysetminus\{j\})  \\
 & \le & \frac{v(S'_{j-1}\cup\{j\}) - v(S'_{j-1})}{b_j}\cdot \frac{(1-U\epsilon)B}{2}  \\
 & \le & \frac{v(S_{j-1}\cup\{j\}) - v(S_{j-1})}{b_j}\cdot \frac{(1-U\epsilon)B}{2}  \\
 & \le & \frac{v(S)}{B}\cdot \frac{(1-U\epsilon)B}{2}  = \frac{(1-U\epsilon)v(S)}{2} \,,
\end{IEEEeqnarray*}
where the last inequality follows from the choice of $b_j$, while the next to last inequality follows from submodularity as $[j-1]=S_{j-1}\subseteq S'_{j-1}$.
Therefore, $v(S'_{j}) \ge (1/2-\epsilon U +\epsilon U /2) v(S)$. 

In any case, we have $v(S'_{j}) \ge (1/2- U\epsilon /2) v(S)$ and thus
\[\frac{v(S_{j}) - v(S_{j-1})}{b_j} \ge \frac{v(S'_{j}) - v(S'_{j-1})}{b_j} \ge \frac{2 v(S'_{j})}{(1-U\epsilon)B} \ge \frac{2 (1/2- U\epsilon /2) v(S)}{(1-U\epsilon)B} \ge \frac{v(S)}{B} \,,\]
which contradicts our assumption about $b_j$.
\end{proof}\smallskip

\begin{lemma}\label{lem:greedy_approximation}
Suppose $v$ is a $(B, \epsilon)$-quasi-monotone submodular function on $A$. 
Let $S$ be the set output by \nameref{fig:alg-2}$(A, v, \mathbf{c}, \beta B)$ and $i^*\in\argmax_{i\in \{j\in A\, |\, c_j\le B\}}v(i)$.
Then $\opt(A,  B) \le \frac{1}{1-\epsilon}\left( \frac{1+\beta}{\beta}v(S)+\frac{1}{\beta}v(i^*)\right)$.
\end{lemma}

\begin{proof}
Like in the proof of Lemma \ref{lem:greedy_budget-feasible} we assume that agents are sorted in descending order with respect to their ratio of marginal value over cost, and that $S=[\ell]$. 
Let $S^*$ be an optimal budget-feasible solution, i.e., $v(S^*)=\opt(A,  B)$. By Theorem \ref{thm:alt_SM} we have 
\begin{IEEEeqnarray*}{rCl}
v(S^*) - v(S) & \le & \sum_{i \in S^* \mysetminus S} (v(S\cup\{i\}) - v(S)) - \sum_{i \in S \mysetminus S^*} (v(S\cup S^*) - v(S \cup S^* \mysetminus \{i\})) \\
 & \le & \sum_{i \in S^* \mysetminus S} c_i\cdot \frac{v(S\cup\{i\}) - v(S)}{c_i} - |S \mysetminus S^*| \left( - \frac{\epsilon}{n} \opt(A,  B)\right) \\
 & \le & \max_{i \in S^* \mysetminus S}\frac{v(S\cup\{i\}) - v(S)}{c_i}\cdot \sum_{i \in S^* \mysetminus S} c_i  +  \epsilon v(S^*) \\
 & = & \frac{v(S\cup\{\ell +1\}) - v(S)}{c_{\ell +1}} \cdot \mathbf{c}(S^* \mysetminus S)  +  \epsilon v(S^*) \\
 & \le & \frac{v(S\cup\{\ell +1\})}{\beta B} \cdot \mathbf{c}(S^*)  +  \epsilon v(S^*) \le \frac{v(S\cup\{\ell +1\})}{\beta}  +  \epsilon v(S^*)\\
 & \le & \frac{1}{\beta}(v(S)+v({\ell +1}))  +  \epsilon v(S^*) \le \frac{1}{\beta}(v(S)+v(i^*))  +  \epsilon v(S^*)\,.
\end{IEEEeqnarray*}
By rearranging the terms we get $\opt(A,  B) \le \frac{1}{1-\epsilon}\left( \frac{1+\beta}{\beta}v(S)+\frac{1}{\beta}v(i^*)\right)$.
\end{proof}



\subsection{Proofs of Theorems \ref{thm:poly_ssm_det} and \ref{thm:poly_wcut_det}}
\label{app:poly_ssm_proofs}
\begin{proofof}{Proof of Theorem \ref{thm:poly_ssm_det}}
We are going to need a slight variant of \nameref{fig:alg-1b-frac}: 
\smallskip

\begin{algorithm}[H]
	\DontPrintSemicolon 
	\NoCaptionOfAlgo
	\algotitle{\textsc{Mech-SM-frac-var}}{fig:alg-1b-frac-var.title}
	Set $A'=\{i\ |\ c_i\le B\}$ and $i^*\in\argmax_{i\in A'}v(i)$ \;
	\vspace{2pt} \If{$\left( \rho +1 +\sqrt{\rho^2 +4\rho +1}\right) \cdot v(i^*) \ge \fopt(A'\mysetminus\{i^*\},   B)$}{\vspace{2pt}\Return $i^*$}
	\Else{\Return \textsc{Greedy-SM}$(A, v, \mathbf{c}, \gamma B/2)$}
	\caption{\textsc{Mech-SM-frac-var}$(A, v, \mathbf{c}, B, \gamma)$} \label{fig:alg-1b-frac-var} 
\end{algorithm}\smallskip 

Recall that $\alpha = (1+\rho)\left( 2+\rho + \sqrt{\rho^2 +4 \rho +1}\right) -1$.  
Now we are ready to state mechanism  \nameref{fig:det-mech-symsm-frac}.
Clearly the mechanism runs in polynomial time.
\smallskip 
 
\begin{algorithm}[H]
	\DontPrintSemicolon 
	\NoCaptionOfAlgo
	\algotitle{\textsc{Det-Mech-SymSM-frac}}{fig:det-mech-symsm-frac.title}
		Set $A'=\{i\ |\ c_i\le B\}$ and $i^*\in\argmax_{i\in A'}v(i)$ \;
		\vspace{2 pt}\If{$\alpha \cdot v(i^*) \ge \fopt(A'\mysetminus\{i^*\},   B)$ \label{line2:det-mech-symsm-frac}}{\vspace{2pt}\Return $i^*$ \label{line:ssm2}} 
		\Else{$S = \nameref{fig:alg-ls}(A, v, \epsilon')$ \;
			\If{$\fopt(S\cap A',  B)\ge \fopt(A'\mysetminus S,  B)$ \label{line:ssm5}}{\vspace{2 pt}\Return \nameref{fig:alg-1b-frac-var}$(S, v, \mathbf{c}_{S}, B, (1-(\alpha+2) \epsilon'))$\label{line:ssm7}}
			\Else{\Return \nameref{fig:alg-1b-frac-var}$(A\mysetminus S, v, \mathbf{c}_{A\mysetminus S}, B, (1-(\alpha+2) \epsilon'))$\label{line:ssm9}}
		}
	\caption{\textsc{Det-Mech-SymSM-frac}$(A, v, \mathbf{c}, B)$} \label{fig:det-mech-symsm-frac} 
\end{algorithm}\smallskip 

The  proof follows  the proof of Theorem \ref{thm:poly_cut_det}. In fact, the monotonicity---and thus truthfulness and individual rationality---of the mechanism follows from that proof and the observation that \nameref{fig:alg-1b-frac-var} is monotone even when $v$ is non-monotone. The latter is due to the monotonicity of \nameref{fig:alg-2} which is straightforward and is briefly discussed in Appendix \ref{app:proof_of_greedy}.

We proceed with the approximation ratio.
If $i^*$ is returned in line \ref{line:ssm2}, then 
\[\alpha \cdot v(i^*) \ge \opt_f(A'\mysetminus\{i^*\},  B) \ge \opt(A'\mysetminus\{i^*\},  B)= \opt(A\mysetminus\{i^*\},  B) \ge \opt(A, B) - v(i^*)\,,\] 
and therefore $\opt(A,  B)\le (\alpha +1) \cdot v(i^*)$. 

On the other hand, if  $X$ is returned by \nameref{fig:alg-1b-frac-var} in line \ref{line:ssm7}, then we have
\[\alpha \cdot v(i^*) < \opt_f(A'\mysetminus\{i^*\}, B) \le \rho \cdot \opt(A'\mysetminus\{i^*\}, B) \le \rho \cdot \opt(A, B).\] Therefore, $v(i^*) <\frac{\rho}{\alpha} \opt(A,   B)$.

At line \ref{line:ssm5} it must be the case $\fopt(S\cap A', B)\ge \fopt(A'\mysetminus S, B)$. Thus,
\begin{IEEEeqnarray*}{rCl}
\rho \cdot \opt(S, B)  & = & \rho \cdot \opt(S\cap A', B) \ge \fopt(S\cap A', B)\ge \fopt(A'\mysetminus S, B)  \\
 & \ge & \opt(A'\mysetminus S, B) = \opt(A\mysetminus S, B) \ge \opt(A, B) - \opt(S, B)\,.
\end{IEEEeqnarray*}
%
Therefore $\opt(S, B) \ge \frac{1}{\rho +1} \opt(A, B)$. Now we need the following lemma about the performance of \nameref{fig:alg-1b-frac-var}.

\begin{lemma}\label{lem:app_D}
For $\eta = \rho +1 +\sqrt{\rho^2 +4\rho +1}$ and any $\epsilon''>0$, there is a sufficiently small $\epsilon'$ so that $\opt(S,  B) \le (\eta +1+\epsilon'') \nameref{fig:alg-1b-frac-var}(S, v, \mathbf{c}_{S}, B, (1-(\alpha+2) \epsilon'))$. Moreover $\epsilon'$ only depends on the constants $\rho$ and $\epsilon''$.
\end{lemma}
\begin{proofof}
\noindent  Let  $\delta = (1-(\alpha+2) \epsilon')/2$. We consider two cases for \nameref{fig:alg-1b-frac-var}.

 If $i^*$ is returned, then 
$\eta \cdot v(i^*) \ge \opt_f(A'\mysetminus\{i^*\},  B) \ge \opt(A'\mysetminus\{i^*\},   B) = \opt(A\mysetminus\{i^*\},   B) \ge \opt(A,   B) - v(i^*)$, 
and therefore $\opt(A,   B)\le (\eta +1) \cdot v(i^*)$. 

On the other hand, if the outcome $X$ of \nameref{fig:alg-2} is returned, then
$\eta \cdot v(i^*) < \opt_f(A'\mysetminus\{i^*\},   B) \le \rho\cdot \opt(A'\mysetminus\{i^*\},   B) \le \rho\cdot \opt(A,   B)$. 
Combining this with Lemma \ref{lem:greedy_approximation} we have 
\[\opt(A,  B) \le \frac{1}{1-\epsilon'}\left( \frac{1+\delta}{\delta}v(X)+ \allowbreak \frac{1}{\delta}\frac{\rho}{\eta}\opt(A,   B)\right) \,,\] 
or equivalently 
\[\opt(A,  B) \le \frac{(1+\delta)\eta}{(1-\epsilon')\delta\eta -\rho}  v(X)\,.\] 

Since $\lim_{\epsilon'\rightarrow 0}\frac{(1+\delta)\eta}{(1-\epsilon')\delta\eta -\rho}  = \frac{3\eta}{\eta -2\rho}$, we have that for sufficiently small $\epsilon'$, $\opt(A,  B) \le \left( \frac{3\eta}{\eta -2\rho}+\epsilon''\right)  v(X) = (\eta +1+\epsilon'') v(X)$, where the calculations for the last equality are the same as in the proof of the approximation ratio of \nameref{fig:alg-1b-frac}.
{\footnotesize  \qed}
\end{proofof}

Now, combining all the above,
we have
\[\opt(A, B)\le (\rho +1)\opt(S, B) \le (\rho +1) \left( \rho +2 +\epsilon'' +\sqrt{\rho^2 +4\rho +1}\right) v(X) \,.\]
For $\epsilon''= \frac{\epsilon}{\rho +1}$ we have  $(\rho +1) \left( \rho +2 +\epsilon'' +\sqrt{\rho^2 +4\rho +1}\right)= \alpha +1 +\epsilon$, as desired. So, the  $\epsilon'$ used in the mechanism is the one we get from the proof of Lemma \ref{lem:app_D} for $\epsilon''= \frac{\epsilon}{\rho +1}$.

The case where $X$ is returned in line \ref{line:ssm9} is symmetric. We conclude that in any case $\opt(A,  B)\le (\alpha +1 +\epsilon) \cdot \nameref{fig:det-mech-symsm-frac}(A, B)$.

It remains to show that the mechanism is budget feasible.
If the winner is $i^*$ in line \ref{line:cut2}, then his payment is $B$. Otherwise, we would like budget-feasibility to follow from the budget-feasibility of \nameref{fig:alg-1b-frac-var} and the observation that the comparison in line \ref{line:cut2} only gives additional upper bounds on the payments of winners from \nameref{fig:alg-1b-frac}. However, the budget-feasibility of \nameref{fig:alg-1b-frac-var} depends on the budget-feasibility of \nameref{fig:alg-2}, and according to Lemma \ref{lem:greedy_budget-feasible} it suffices to have $v(X)\ge \frac{1}{\alpha+2} \cdot \opt(A,  B)$, where $X$ is the output of \nameref{fig:alg-2}. This, however, follows from the approximation ratio for $\epsilon\le 1$. Thus the mechanism is budget-feasible. 
\qed
\end{proofof}
\bigskip

\begin{proofof}{Proof of Theorem  \ref{thm:poly_wcut_det}}
The proof is identical with the proof of Theorem \ref{thm:poly_ssm_det} with the exception of the analysis of the approximation ratio which borrows the ideas of the proof of Theorem \ref{thm:poly_cut_det}. So we only focus on the approximation ratio.

If $i^*$ is returned in line \ref{line:ssm2}, then $\opt(A,  B)\le (\alpha +1) \cdot v(i^*)$ like before. 

Otherwise, using  Theorem \ref{lem:bwmcut},  we have
\[\alpha \cdot v(i^*) < \opt_f(A'\mysetminus\{i^*\}, B) \le 4 \cdot \opt(A'\mysetminus\{i^*\}, B) \le 4 \cdot \opt(A, B).\] Therefore, $v(i^*) <\frac{4}{\alpha} \opt(A,  B)$  and for the remaining steps of the mechanism  we can use  Theorem \ref{lem:bwmcut} with factor $\rho' = 2+8/\alpha$. 

Without loss if generality, assume  that  $X$ is returned by \nameref{fig:alg-1b-frac-var} in line \ref{line:ssm7}. 
At line \ref{line:ssm5} it must be the case $\fopt(S\cap A', B)\ge \fopt(A'\mysetminus S, B)$. Thus,
following the previous analysis, $\opt(S, B) \ge \frac{1}{\rho' +1} \opt(A, B)$. 

Combining this inequality with Lemma \ref{lem:app_D} for $\rho'$ we get that for $\epsilon''= \frac{\epsilon}{\rho' +1}$  there is an $\epsilon'$ such that
\[\opt(A, B)\le  (\rho' +1) \left( \rho' +2 +\epsilon'' +\sqrt{\rho'^2 +4\rho' +1}\right) v(X) = (\alpha +1 +\epsilon)\cdot v(X)\,.\]
This $\epsilon'$ is to be used in the mechanism. Showing that $\alpha=26.245$ works is only a matter of calculations, and for $\epsilon\le 1/200$ we get a 27.25-approximate solution.
\qed
\end{proofof}

\end{document}